\documentclass[journal,10pt]{IEEEtran}

\usepackage[nocompress]{cite}
\usepackage{url}
\usepackage[pdftex]{graphicx}
\usepackage[cmex10]{amsmath}
\usepackage[linesnumbered,lined,vlined,ruled,commentsnumbered]{algorithm2e}
\usepackage{array}
\usepackage[caption=false,font=footnotesize,labelfont=sf,textfont=sf]{subfig}
\usepackage{fixltx2e}
\usepackage{stfloats}
\usepackage{multirow}
\usepackage{color}
\usepackage{amssymb}
\usepackage[english]{babel}
\usepackage{amsthm}
\usepackage{bm}

\usepackage[T1]{fontenc}% optional T1 font encoding
\usepackage[numbers,sort&compress]{natbib}
\usepackage{url}
\usepackage{booktabs}

\newtheorem{theorem}{Theorem}
\newtheorem{lemma}{Lemma}

\newtheorem{corollary}{Corollary}

\newtheorem{game}{Game}

\def\S{\mathcal{S}}
\def\eq{\triangleq}
\def\p{\boldsymbol{p}}
\def\r{\boldsymbol{r}}

\begin{document}
\title{Reputation and Pricing Dynamics in Online Markets}

\author{Qian~Ma,~\IEEEmembership{Member,~IEEE,}
        Jianwei~Huang,~\IEEEmembership{Fellow,~IEEE,}
        Tamer~Ba\c{s}ar,~\IEEEmembership{Life Fellow,~IEEE,}
        Ji~Liu,~\IEEEmembership{Member,~IEEE,}
        and~Xudong~Chen,~\IEEEmembership{Member,~IEEE}% <-this % stops a space
\thanks{Q. Ma is with the School of Intelligent Systems Engineering, Sun Yat-sen University. E-mail: maqian25@mail.sysu.edu.cn.}% <-this % stops a space
\thanks{J. Huang is with the School of Science and Engineering, The Chinese University of Hong Kong, Shenzhen, Shenzhen 518172, China, and the Shenzhen Institute of Artificial Intelligence and Robotics for Society, Shenzhen 518129, China. E-mail: jianweihuang@cuhk.edu.cn.  (\emph{Corresponding author: Jianwei Huang.})}% <-this % stops a space
\thanks{T. Ba\c{s}ar is with the Department of Electrical and Computer Engineering, University of Illinois at Urbana-Champaign. E-mail: 	
basar1@illinois.edu.}% <-this % stops a space
\thanks{J. Liu is with the Department of Electrical and Computer Engineering, Stony Brook University. E-mail: 	
ji.liu@stonybrook.edu.}% <-this % stops a space
\thanks{X. Chen is with the Department of Electrical, Computer and Energy Engineering, University of Colorado Boulder. E-mail: xudong.chen@colorado.edu.}% <-this % stops a space
\thanks{This work is supported by the National Natural Science Foundation of China under Grant 62002399, the Shenzhen Institute of Artificial Intelligence and Robotics for Society, and the Presidential Fund from the Chinese University of Hong Kong, Shenzhen.}
\thanks{An earlier version of this work was presented at the 11th Workshop on the Economics of Networks, Systems and Computation (NetEcon) (in conjunction with ACM SIGMETRICS 2016), Juan-les-Pins, France, June 2016, and appeared as a one-page summary \cite{NetEcon} in its Proceedings.}
}

\maketitle

\begin{abstract}
We study the economic interactions among sellers and buyers in online markets. 
In such markets, buyers have limited information about the product quality, but can observe the sellers' reputations which depend on their past transaction histories and ratings from past buyers. 
Sellers compete in the same market through pricing, while considering the impact of their heterogeneous reputations. 
We consider sellers with limited as well as unlimited capacities, which correspond to different practical market scenarios. 
In the unlimited seller capacity scenario, buyers prefer the seller with the highest reputation-price ratio. 
If the gap between the highest and second highest seller reputation levels is large enough, then the highest reputation seller dominates the market as a monopoly. 
If sellers' reputation levels are relatively close to each other, then those sellers with relatively high reputations will survive at the equilibrium, while the remaining relatively low reputation sellers will get zero market share. 
In the limited seller capacity scenario, we further consider two different cases. 
If each seller can only serve one buyer, then it is possible for sellers to set their monopoly prices at the equilibrium while all sellers gain positive market shares; if each seller can serve multiple buyers, then it is possible for sellers to set maximum prices at the equilibrium. 
Simulation results show that the dynamics of reputations and prices in the longer-term interactions will converge to stable states, and the initial buyer ratings of the sellers play the critical role in determining sellers' reputations and prices at the stable state.
\end{abstract}


\section{Introduction}

\subsection{Background and Motivation}

\IEEEPARstart{T}{he} emergence of online markets has made it possible for geographically separated sellers and buyers to conduct transactions with each other with small transaction costs. 
Online markets such as Amazon, eBay, and Taobao (the largest online market in China) are becoming increasingly important in our daily lives. 
For example, the amount of Taobao sales on the single day of Nov. 11, 2017 reached 26 billion US dollars. 
Some of the online markets correspond to the online sharing economy platforms \cite{CollaborativeConsumption}, which facilitate online peer-to-peer fee-based resource sharing between resource sellers (owners) and buyers. 
On these sharing platforms, sellers earn profits by allowing others to access their under-utilized (online or offline) resources, and buyers obtain resources at cheaper prices than through conventional approaches \cite{MotivationsSharing}. 
The increase of consumer awareness and development of online platforms make online sharing economy increasingly popular, with many successful examples such as Airbnb for room sharing and Uber for car sharing.

In online markets, the quality of products (e.g., quality of products on Amazon and comfort level of Airbnb rooms) has great impact on buyers' experiences. 
However, buyers often have limited information about the product quality at the time of a transaction. 
This is because buyers often cannot try the products before their purchase. 
In a market with many small sellers, a buyer often needs to purchase products from a seller whom he has never or seldom transacted with. 
One way to estimate the product quality is to observe a seller's reputation, which depends on the number of transactions completed by the seller and the review ratings received from past buyers \cite{Manipulation}.

A seller's reputation also affects the seller's pricing strategy. 
A seller with a higher reputation naturally attracts more buyers and can set a higher price, which in turn encourages the seller to provide better products (with potentially higher costs) and keep the reputation high \cite{RepuSys}. 
A seller with a lower reputation, however, is less attractive to buyers, and will have a significant disadvantage when competing with other sellers.

\subsection{Model and Problem Formulation}

In this work, we analyze the long-term dynamics of an online market, where different sellers sell products in the same category (e.g., TV sellers \cite{TCLCompetition} on Amazon or luxurious apartment owners on Airbnb) and can choose different prices. 
Sellers are heterogeneous in terms of their initial reputation levels, which depend on the number of completed transactions and the past buyer ratings. 
New buyers arrive at the online market according to a stochastic process, and observe sellers' reputations and prices upon arrival. 
Each buyer chooses a seller and determines the corresponding purchasing amount to maximize the buyer's payoff. 
Figure \ref{fig:model} illustrates such an online market with three sellers and randomly arriving buyers.

In this paper, we would like to answer the following key questions considering two scenarios depending on sellers' capacities, i.e., the amount of products or services that sellers can provide:
\begin{itemize}
  \item \emph{Unlimited Capacity Scenario}: If sellers have unlimited capacities, how should a buyer select among sellers with heterogeneous reputations and prices to maximize his payoff? How should sellers set their prices to maximize their own profits by taking their reputations into consideration? % either competitively or cooperatively,
  \item \emph{Limited Capacity Scenario}: How does limited capacity change the behaviors of the buyers and the sellers? 
\end{itemize}

\begin{figure}[t]
\vspace{-3mm}
\centering
\includegraphics[width=0.32\textwidth]{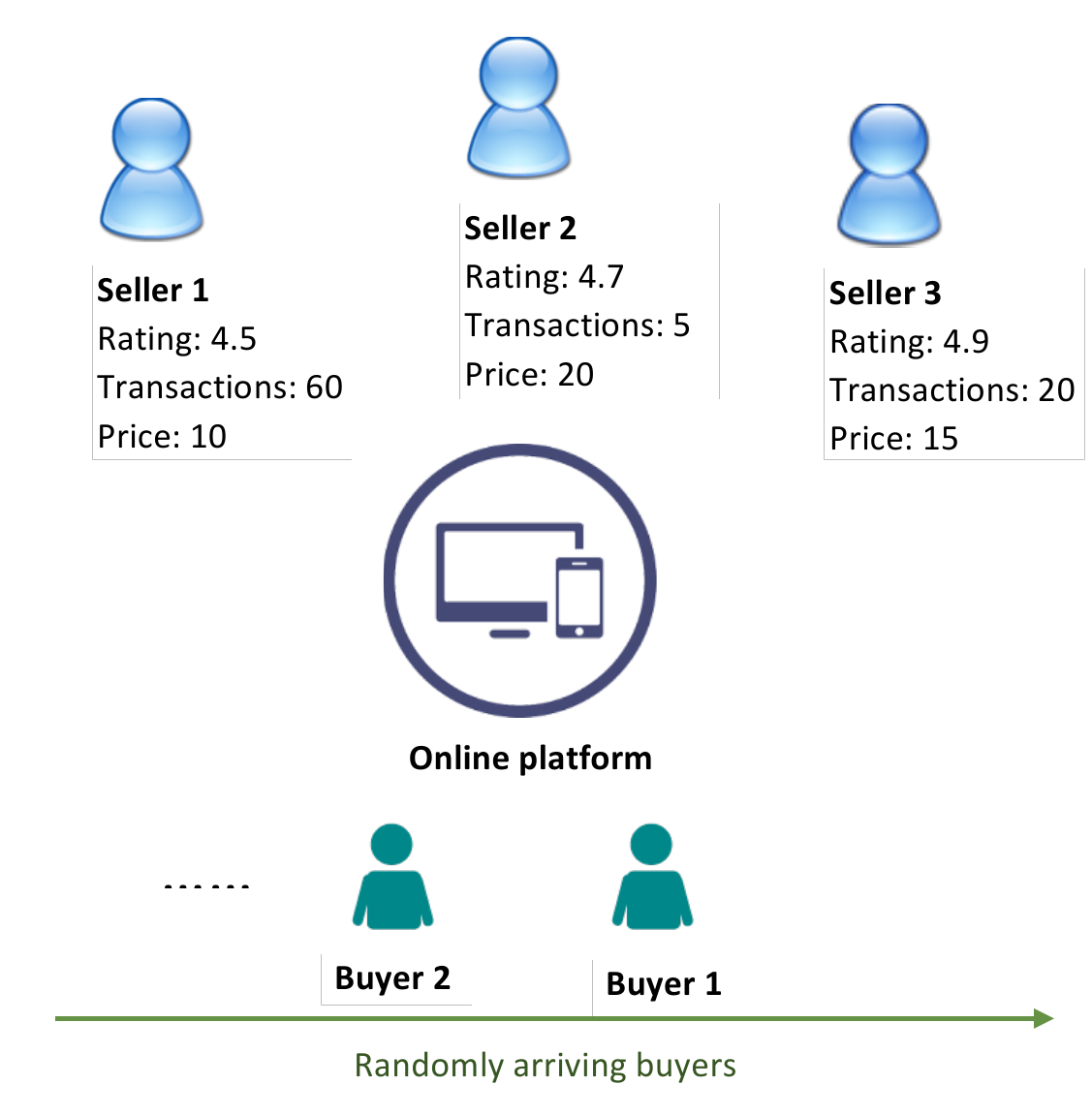}
\vspace{-3mm}
\caption{An online market with 3 sellers and randomly arriving buyers
} \label{fig:model}
\vspace{-3mm}
\end{figure}

\subsection{Solutions and Contributions}

We consider an infinite time horizon model as in Figure 2(a), where the time is divided into many time slots. 
Each time slot can be one day for Amazon or one week for Airbnb. 
We model the interactions among sellers and buyers as a dynamic game, where Figure 2(b) shows that in each time slot they play a two-stage multi-leader-multi-follower game \cite{BasarDynamicGame}.
Specifically, at the beginning of each time slot, sellers announce their unit prices for their products.
Buyers arrive at the online market according to a stochastic process, and each newly arrived buyer decides which seller to choose and what amount to buy, based on the announced prices and publicly observable seller reputations in that time slot.

\begin{figure*}
\vspace{-3mm}
\centering
\subfloat[]{
\label{fig:TimeSlotted}
\begin{minipage}[t]{0.53\textwidth}
\centering
\includegraphics[width=1\textwidth]{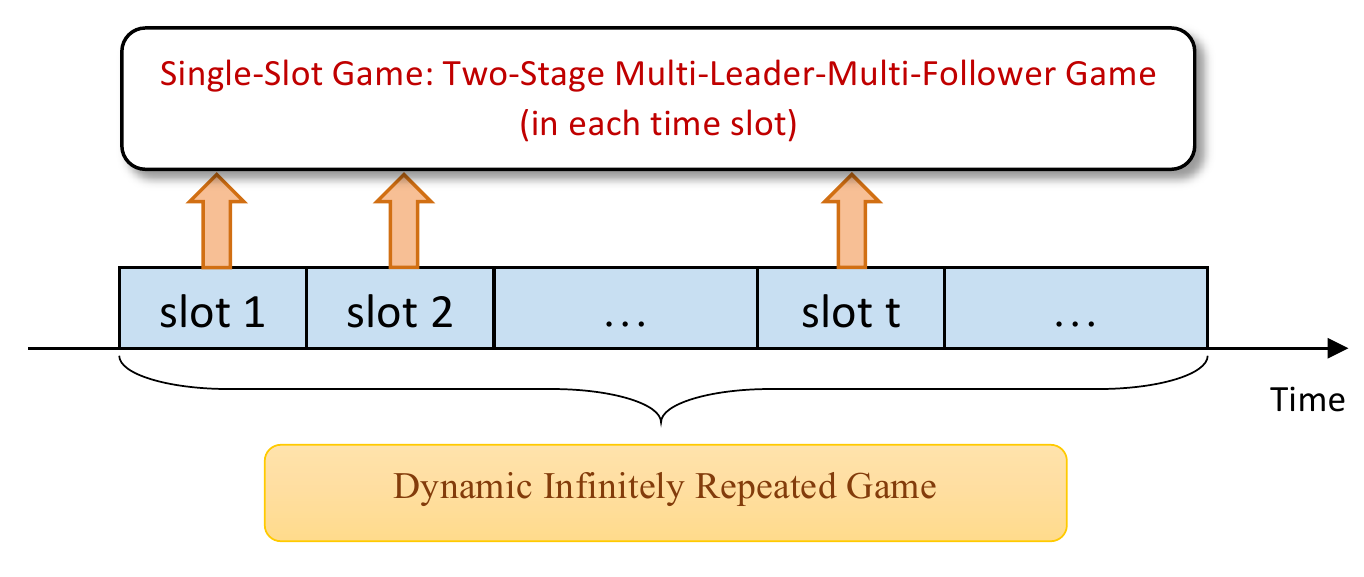}
\end{minipage}
}
\subfloat[]{
\label{fig:TwoStage}
\begin{minipage}[t]{0.47\textwidth}
\centering
\includegraphics[width=0.7\textwidth]{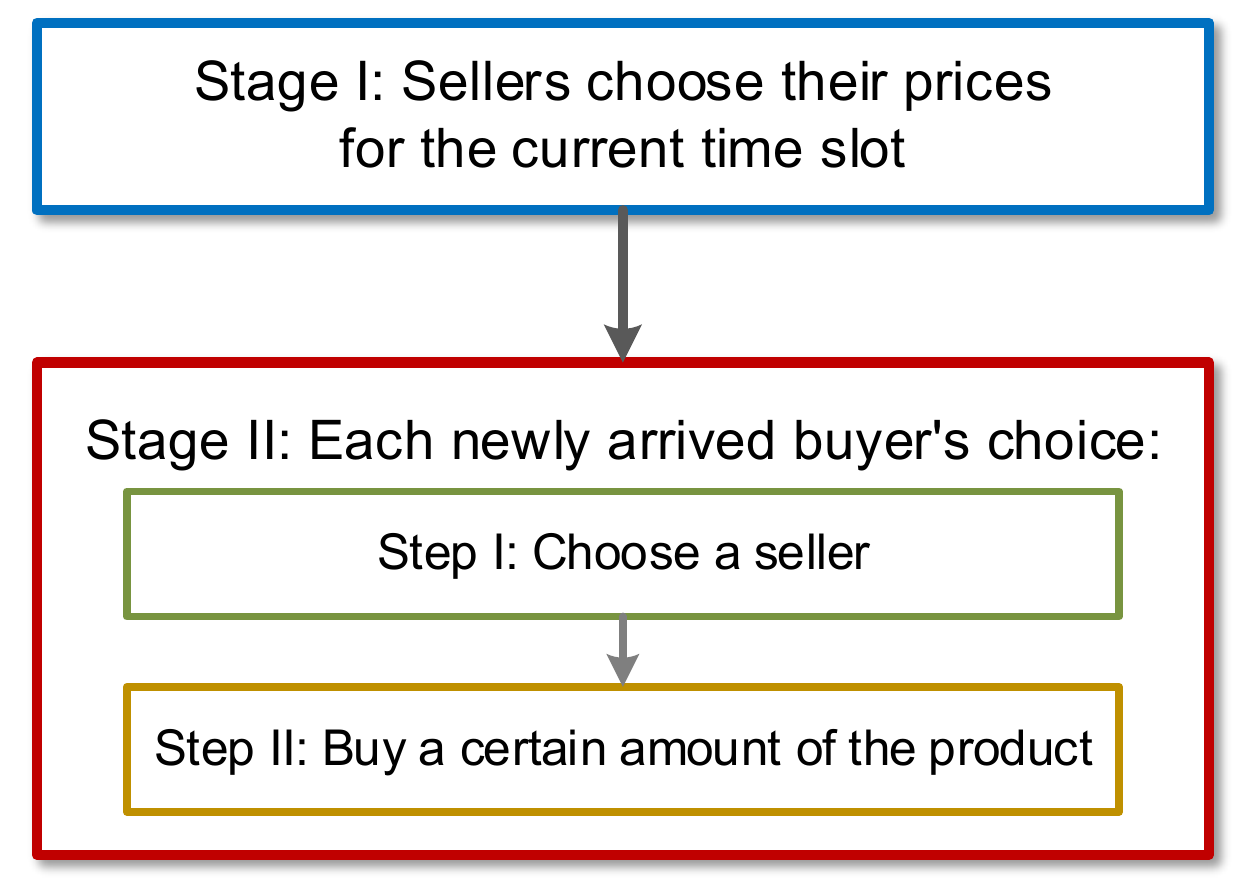}
\end{minipage}
}
\caption{Interactions between sellers and buyers: (a) the dynamic infinitely repeated game with many time slots; (b) the two-stage multi-leader-multi-follower game in each time slot.}
\vspace{-3mm}
\end{figure*}

As far as we know, this is the first work that provides a comprehensive economic analysis of seller competition in an online market considering heterogeneous reputations. 
We summarize our key contributions as follows.

\begin{itemize}

\item \emph{Unlimited Capacity Scenario Analysis:}
For the scenario where each seller has unlimited capacity, we show that if the gap between the highest and the second highest seller reputation levels is large enough, then the market becomes a monopoly market dominated by the highest reputation seller. 
When the sellers' reputation levels are relatively close, then those sellers with relatively high reputations will survive at the equilibrium, while the remaining low reputation sellers will be driven out of the market.
Furthermore, as time evolves, it is possible for the market to change from a multi-seller market to a single seller monopoly.

\item \emph{Limited Capacity Scenario Analysis:}
In the limited capacity scenario, if each seller can serve one buyer, then sellers can set monopoly prices at equilibrium and all sellers gain positive market shares; if each seller can serve multiple buyers, then sellers can set maximum prices at equilibrium.

\item \emph{Simulation Results:}
Simulation results show that the dynamics of reputations and prices at the equilibrium of the long-term interactions for both sellers with unlimited capacity and sellers with limited capacity will converge to stable states, and the initial buyer ratings of the sellers play a critical role in determining  sellers' reputations and prices under the stable state.

\end{itemize}

The rest of the paper is organized as follows. 
We provide literature review in Section \ref{sec:literatures}.
We describe the system model and present the problem formulation in  Section \ref{sec:model}. 
We analyze the unlimited seller capacity scenario in Section \ref{sec:unlimited}. 
In Section \ref{sec:limitedone}, we analyze the scenario where each seller has limited capacity and can only serve one buyer in a time slot. 
In Section \ref{sec:limitedmulti}, we analyze the scenario where a seller has limited capacity and can serve multiple buyers in a time slot. 
We present simulation results in Section \ref{sec:simu}, and conclude in Section \ref{sec:conc}. 
%Some proofs have been relegated to the supplementary Appendix in the online submission system.

\section{Literature Review}\label{sec:literatures}

\subsection{Competition in Traditional (Offline) Markets}

There is a rich body of literature on competition of traditional (offline) markets, e.g., competition over prices with private capacity information \cite{privatecapa}, competition over capacities and prices \cite{Ozdaglar2}, and competition targeted to buyers embedded in a social network \cite{Ozdaglar1}. 
However, the above results for offline markets have not considered the impact of sellers' reputation. 
Another body of work (e.g., \cite{Jadbabaie}) studied how buyers obtain knowledge of the product through the word-of-mouth communications over a social network. 
In that case, only a few number of friends' opinions affect the buyer's knowledge of the seller's product, and the impact of a friend's opinion depends on the social relationship strength between the friend and the buyer. 
However, in online markets, the reputation is based on all buyers' reviews (ratings), which reveal more product information and make the transactions between buyers and sellers who are not connected in social networks possible \cite{Trust}.

\subsection{Reputation of Online Markets}
There has been some research on reputation systems \cite{RepuSys} for online markets, such as those focusing on reputation characterization \cite{Trust}, empirical studies on the impact of reputation on behaviors of sellers and buyers \cite{negative,DynaRepu,experiment,MIS,Pennies}, reputation accumulation and manipulation \cite{XieHongPE,Manipulation,sybilproof}, and incentive mechanism design for feedback \cite{unfair,honest,Johari2}.
There is no existing work that theoretically characterizes how reputation explicitly affects sellers' pricing competition and buyers' choices. 
To the best of our knowledge, this paper is the first one that theoretically studies the pricing competition among sellers with heterogeneous reputations and explores the impact of seller reputation on buyers' choices.

\subsection{Online Sharing Economy Platforms}
Research results regarding online sharing economy platforms only emerged recently (e.g., \cite{MotivationsSharing,NetEcon15,Johari1,Airbnb}). 
Hamari \emph{et al.} in \cite{MotivationsSharing} studied people's motivations to participate in online sharing economy platforms. 
Benjaafar \emph{et al.} in \cite{NetEcon15} studied how an online platform in sharing economy could maximize profit or social welfare. 
Banerjee \emph{et al.} in \cite{Johari1} studied dynamic pricing for online ride-sharing platforms, focusing on achieving a robust system performance considering the stochastic dynamics of the marketplace. 
Zervas \emph{et al.} in \cite{Airbnb} studied the ratings at Airbnb and concluded that the average rating on Airbnb was higher than that on TripAdvisor. 
However, there is no prior work focusing on the impact of reputation on the interactions among sellers and buyers where sellers have limited capacities, and this paper is the first attempt in this direction to the best of our knowledge.

\section{System Model}\label{sec:model}

We study an online market consisting of a set $\mathcal{S}=\{1,2,\ldots,S\}$ of sellers who sell the same category of products. 
The time is slotted, and we assume that buyers arrive at the online market randomly, with the probability of having $k$ buyers arriving in each time slot being $P_k$. 
Sellers only know the the distribution of $\{P_k\}$ with $\sum_{k=1}^\infty P_k = 1$, without observing the actual value of $k$ at the beginning of a time slot (when the sellers make pricing decisions). 
We further assume that the buyers are homogeneous in evaluating each seller's reputation and deriving their own satisfaction levels, and hence there is no need to index the buyers differently.

In the following, we first describe the interactions between sellers and buyers, and then present a model of the seller reputation. 
Finally, we describe three different capacity scenarios to be analyzed in Sections \ref{sec:unlimited}, \ref{sec:limitedone}, and \ref{sec:limitedmulti}.

\subsection{Interactions between Sellers and Buyers}

Sellers and buyers interact at two time scales: in each time slot and in the long term.

Figure 2(b) illustrates the interactions between sellers and buyers in each single time slot, which can correspond to one day for Amazon or one week for Airbnb.  
At the beginning of a time slot, each seller $s\in\S$ decides the unit price $p_s$ for his product to maximize his profit, considering the competition among all sellers with heterogeneous reputations $\boldsymbol{r} = \{r_i: \forall i\in\S\}$. 
If sellers are myopic and only care about their profits in the current slot, we can model sellers' pricing decisions through a static seller competition game model. 
Buyers arrive at the online market randomly during the time slot, and each newly arrived buyer decides on which seller $s\in\S$ to choose and the amount $x_s$ to buy from the chosen seller to maximize the buyer's payoff,\footnote{The range of the continuous consumption amount $x_s$ is the set of non-negative real numbers. The continuous modeling of the purchasing amount is an approximation of the discrete modeling for examples such as books and Airbnb rooms. In some cases it can be an accurate reflection when the product or service is divisible. For example, for the online Amazon Elastic Compute Cloud (Amazon EC2), the amount of time to rent computing resources is a continuous variable.} based on the prices announced by sellers and the available seller reputations during the time slot.

Figure 2(a) captures the seller-buyer interactions in the long term consisting of an infinite number of time slots. 
When the sellers are far-sighted and want to maximize their long-term profits, then they need to consider the repeated interactions among sellers, which can be modeled as a dynamic infinitely repeated game. 
Even in the infinite time horizon, we still assume that each buyer is interested only in his payoff in a single time slot \cite{shortbuyer}, as  buyers are usually bounded rational due to limited computation capabilities \cite{boundedrationality}. 
As a buyer often has limited demand and is not likely to purchase the same product repeatedly from the same seller (e.g., a buyer will not buy a Kindle from Amazon every day or book rooms on Airbnb every week), it is difficult for him to know the true quality of each seller through his own purchase experiences. 
Hence, it would be natural for a buyer to choose among sellers based on sellers' current prices and public reputations up to that point.

\subsection{Seller Reputation}

A seller's reputation, which depends on the number of completed transactions and past buyer ratings, plays an important role in the buyers' decisions.
For a seller $s$, we let $X_s[t]$ denote the number of completed transactions up to the beginning of time slot $t$, and let $\omega_s \in [0,1]$ denote the past buyers' rating. 
We assume that the rating $\omega_s ~(\forall s\in\S)$ changes much slower compared with the change of $X_s[t]$, and hence we approximate the rating as fixed in a relatively long period of time (i.e., the period of interest in our decision model). 
Such an assumption has been widely adopted and verified in the literature on e-commerce marketplaces (e.g., \cite{ratingunchanged,ratingunchanged2,ratingunchanged3}). 
Hence, the reputation of seller $s$ in time slot $t$, denoted by $r_s[t]$, can be calculated as a function of $X_s[t]$ and $\omega_s$.

In this paper we use the P\'{o}lya urn model \cite{CumulativeAdvantage},\footnote{There can be other formulations for $r_s$. For example, the reputation of seller $s$ can be the number of transactions with high ratings (or good reviews) out of his total number of transactions, which is simpler than the P\'{o}lya urn model. However, such a formulation does not capture the cumulative advantage of past transactions.} which has been widely applied to capture the competition dynamics among various sellers considering the intrinsic competitiveness and the cumulative advantage, to model the reputation:\footnote{Our theoretical analyses in Sections \ref{sec:unlimited}, \ref{sec:limitedone}, and \ref{sec:limitedmulti}, except the analysis of reputation evolution dynamics in \eqref{eq:repdyn} and the simulation results on the reputation dynamics in Section \ref{sec:simu},  only depend on the value of the reputation $r_s$ rather
than the values of the rating $w_s$ or the transaction number $X_s[t]$, and hence are general and do not rely on the specific form of \eqref{eq:reputation}.} 
\begin{equation}\label{eq:reputation}
\displaystyle r_s[t]=\frac{\omega_s X_s[t]}{ \sum_{i\in\S} \omega_i X_i[t]}. 
\end{equation}

Intuitively, a higher rating $\omega_s$ and a larger number of completed transactions $X_s[t]$ lead to a higher level of reputation $r_s[t]$ for seller $s$ compared with other sellers. 
First, a higher rating $\omega_s$ implies that users who have used the product experienced a better product quality. 
The rating $\omega_s$ hence reflects the \emph{intrinsic competitiveness}, i.e., the inherent ability of seller $s$ to win the competition among all sellers \cite{CumulativeAdvantage}. 
Second, a larger number of completed transactions $X_s[t]$ implies that the product is more popular among buyers. 
Hence $X_s[t]$ reflects the \emph{cumulative advantage}, i.e., the impact of completed transactions on promoting more transactions in future competitions, reflecting the dictum  ``the rich gets richer'' \cite{CumulativeAdvantage}.

The key feature of the model in \eqref{eq:reputation} is that different sellers' reputation levels are interdependent. 
Under fixed values of $\omega_s$ and $X_s[t]$, seller $s$' reputation is higher when the other sellers are ``weaker''. 
Table \ref{table:example} illustrates this point by computing a seller $A$'s reputation in two different markets. 
As seller B is weaker than seller C, seller A's reputation is higher in Market I than in Market II.

\begin{table}[t]
%\vspace{-4mm}
\newcommand{\tabincell}[2]{\begin{tabular}{@{}#1@{}}#2\end{tabular}}
\centering
\caption{Two Markets with Different Sellers}
\begin{tabular}{|c|c|}
\hline
\textbf{Market I}  &  \tabincell{c}{~Seller A: $\omega_A=4.5, X_A[t]=100$, $r_A[t]=~$\textbf{0.9146} \\ Seller B: $\omega_B=4.2, X_B[t]=10, r_B[t]=0.0854$ }  \\
\hline
\textbf{Market II}  & \tabincell{c}{~Seller A: $\omega_A=4.5$, $X_A[t]=100$, $r_A[t]=~$\textbf{0.0841} \\ ~~Seller C: $\omega_C=4.9, X_C[t]=1000, r_C[t]=0.9159$}  \\
\hline
\end{tabular}
\label{table:example}
%\vspace{-3mm}
\end{table}

\subsection{Seller Capacities: Three Scenarios}

We consider three market scenarios, corresponding to different assumptions of a seller's capacity in different practical online markets.

The first scenario is the \emph{unlimited capacity scenario}, where each seller has unlimited capacity and hence can always satisfy the buyers' demands. 
This can be a good approximation for the case where sellers have enough (although finite) production capacity, such as major TV sellers (e.g., TCL and Samsung) on Amazon. 
Furthermore, the analysis for the unlimited capacity scenario serves as a benchmark for the limited capacity scenario. 

The second scenario is the \emph{limited capacity and one buyer per seller scenario}, where each seller has limited capacity and can serve only one buyer in each time slot.  
This is motivated by practical online sharing economy platforms such as Airbnb, where a host of one apartment typically only serves one buyer (such as one traveling family) at any given time. 

The third scenario is the \emph{limited capacity and multiple buyers per seller scenario}, where each seller has limited capacity and can serve multiple buyers in each time slot. 
This is the case of small sellers on Amazon.

\section{Unlimited Capacity Scenario}\label{sec:unlimited}

In this section, we analyze the scenario where each seller has unlimited capacity. 
In the following, we first analyze the single-slot game (the two-stage game in a single time slot), and then analyze the infinite-horizon dynamic game.

\subsection{Single-Slot Game Analysis}\label{sec:unlimitedSingle}

Since the single-slot game as shown in Figure 2(b) is a Stackelberg game, we first analyze the buyers' purchasing decision in Stage II, and then study the sellers' pricing decisions in Stage I. 
Since we focus on a generic time slot here, we will suppress the time index $t$ in Section \ref{sec:unlimitedSingle}. 
We will bring the index $t$ back in Section \ref{subsec:DGUC}.

\subsubsection{Buyers' Purchasing Decisions in Stage II}\label{sec:buyer}

We first introduce a buyer's utility and payoff functions. 
Then we derive the buyer's optimal purchasing decision (at a particular seller) and optimal seller selection decision as the result of his payoff maximization. 

\emph{Buyer utility:} 
We assume that buyers are homogeneous in evaluating each seller's reputation and deriving their own satisfaction levels, and hence there is no need to index the buyers differently. 
We use $u_{s}(x_{s})$ to denote a buyer's expected\footnote{The expected utility is the buyer's ex ante utility before buying the product, which can be different from his perceived utility after buying the product, due to the gap between reputation and real quality.} utility (satisfaction) achieved by buying an amount $x_{s}$ of products from seller $s$. 
A good reputation leads to an increase in the buyer's expected valuation of the products \cite{ReputationSocialRela}.
Following the common assumption of diminishing marginal returns \cite{concaveutility}, we assume that $u_{s}(x_{s})$ is increasing and concave. 
More specifically, we will adopt the following utility function that captures the effect of reputation $r_s$ of seller $s$: 
\begin{equation}\label{eq:utility}
u_{s}(x_{s})=\rho\log(1+r_{s}x_{s}),
\end{equation}
where $\rho>0$ is the buyers' homogeneous product evaluation parameter, which describes how important the product is to a buyer. 
The value of $\rho$ is independent of the sellers, as all sellers sell the same category of product \cite{Tamer}.

\emph{Buyer payoff:} 
When a buyer purchases an amount $x_s$ of the product from seller $s$ at a price $p_s$, the buyer's payoff is the difference between his utility and the payment,\footnote{Note that the buyer's utility depends on the reputation, which is a function of the past price history (not including the current price). Hence, the current price only affects the buyer's payoff in the payment.} i.e., 
\begin{equation}\label{eq:payoff}
v_{s}(x_{s},p_s)=u_{s}(x_{s})-p_sx_{s}.
\end{equation}

\emph{Buyer's optimal consumption:} 
When purchasing from a particular seller $s$, a buyer's optimal purchase decision (consumption amount) can be calculated as follows:
\begin{lemma}\label{lemma:Optx1}
If a buyer chooses seller $s$, then his optimal consumption amount at seller $s$ is: 
\begin{equation}\label{eq:Optx1}
x_s^\ast(p_s)=\max\left \{\frac{\rho}{p_s}-\frac{1}{r_s},0 \right \}.
\end{equation}
\end{lemma}

\begin{proof}
See Appendix \ref{app:lemma_Optx1}. % in the supplementary material.
\end{proof}

When the buyer's optimal consumption from a seller $s$, $x_s^\ast(p_s)$, is positive (i.e., $\rho/p_s - 1/r_s >0$), $x_s^\ast(p_s)$ increases with the reputation $r_s$ and decreases with the price $p_s$. 
In this case, the buyer's optimal payoff achieved by purchasing from seller $s$ is: 
\begin{equation}\label{eq:OPTpayoff}
v_{s}(x_{s}^\ast(p_s),p_s)=\rho \log\left( \rho \frac{r_{s}}{p_s} \right)+\frac{p_s}{r_{s}}-\rho, 
\end{equation}
which monotonically increases in $\varphi_s \eq \frac{r_s}{p_s}$, the \emph{reputation-price ratio}. 
When $x_s^\ast(p_s) = 0$, then $v_s(x_s^\ast(p_s),p_s) = 0$.

\emph{Buyer's optimal seller choice:} 
Finally, the buyer will choose a seller that results in maximum optimal payoff: 
\begin{equation}\label{eq:buyerObj}
\max_{s\in\S} ~\max_{x_{s}\geq 0}~ v_{s}(x_{s},p_s). 
\end{equation}
If multiple sellers yield the same maximum payoff, we assume that the buyer will randomly choose one of them with equal probability. 

The above discussions imply that a buyer's decision is affected by the prices of all sellers. 
Let us define $p_{-s} = (p_i, \forall i\in\mathcal{S}, i\neq s)$ as the prices of sellers except seller $s$ in the market. 
Hence we write the price vector of all sellers as $\boldsymbol{p}=(p_s: \forall s\in\mathcal{S})=(p_s, p_{-s})$, and a buyer's optimal consumption amount at seller $s$ after considering seller selection is $x_{s}^{\rm sl\ast}(p_s,p_{-s})$.\footnote{The superscript ``sl'' stands for seller selection.} 

From \eqref{eq:OPTpayoff}, we know that a buyer's optimal payoff when choosing seller $s$ increases with the seller's reputation-price ratio $\varphi_s$. 
Hence, we have the following result.

\begin{lemma}\label{lemma:sellerSelection}
In Stage II, each buyer will choose a seller with the highest reputation-price ratio, i.e., $s^{\ast}\in \arg \max_{i\in\S}\varphi_i$. 
\end{lemma}

When multiple sellers have the same highest ratio $\max_{i\in\S}\varphi_i$, we denote the set of such sellers as $\mathcal{S}_r=\{s: \varphi_s=\max_{i\in\S}\varphi_i\}$ with a size $S_r$. 
Due to buyers' random choices, each seller in this set will have a positive expected demand.

From Lemmas \ref{lemma:Optx1} and \ref{lemma:sellerSelection}, we have the following result. 
\begin{lemma}\label{lemma:OptCon}
Given a price vector $\boldsymbol{p}$, the buyer's expected consumption amount at each seller $ s \in \S$ after considering seller selection is 
\begin{equation}\label{eq:buyerOpt}
x_{s}^{\rm sl\ast}(p_s,p_{-s})=
\left\{
\begin{aligned}
& \frac{1}{S_r} \left( \frac{\rho}{p_s}-\frac{1}{r_{s}}\right), \mbox{ if } \varphi_s= \max_{i\in\S} \varphi_i, \varphi_s> \frac{1}{\rho};
\\%[-2pt]
& 0, \mbox{ otherwise}.
\end{aligned}
\right.
\end{equation}
\end{lemma}

If $k$ buyers arrive at the online market in a time slot, then the expected total consumption amount at seller $s$ is: 
\begin{equation}\label{eq:qsNES}
q_{s}^\ast(p_s,p_{-s})=
\left\{
\begin{aligned}
& \frac{k}{S_r}\left( \frac{\rho}{p_s}-\frac{1}{r_{s}}\right), \mbox{ if } \varphi_s= \max_{i\in\S} \varphi_i, \varphi_s> \frac{1}{\rho};
\\%[-2pt]
& 0, \mbox{ \emph{otherwise}}.
\end{aligned}
\right.
\end{equation}

Table \ref{table} lists some key notations defined here and in the rest of the paper.

\begin{table}[t]
\newcommand{\tabincell}[2]{\begin{tabular}{@{}#1@{}}#2\end{tabular}}
\centering
\caption{\textsc{Key Notations}}
\begin{tabular}{|c||p{6.5cm}|}
\hline
\textbf{Symbol}   & \textbf{Physical Meaning}   \\
\hline
$\mathcal{S}$  & The set of sellers, $\mathcal{S}=\{1,2,\ldots,S\}$   \\
\hline
$\omega_s$  & The rating of seller $s$  \\
\hline
$X_s[t]$  & The number of completed transactions of seller $s$ up to time slot $t$  \\
\hline
$r_s[t]$  & The reputation of seller $s$ in time slot $t$ \\
\hline
$q_s[t]$  & The number of completed transactions of seller $s$ in the single time slot $t$  \\
\hline
$p_s$  & The unit price announced by seller $s$ \\
\hline
$c$  & Sellers' common marginal cost \\
\hline
$x_s$ & The amount of products that a buyer buys at seller $s$ \\
\hline
$\rho$ & The buyers' common product evaluation parameter \\
\hline
$\varphi_s$ & The reputation-price ratio of seller $s$  \\
\hline
$r_{\max}$  & The highest level of reputation in the network \\
\hline
$r_{\sec}$  & The second highest level of reputation in the network \\
\hline
$\mathcal{S}_r$  & The set of sellers with the highest reputation-price ratio  \\
\hline
$\varepsilon$  & A sufficiently small positive number \\
\hline
$\varphi^{\rm MS}$  & The reputation-price ratio under the multi-seller strategy \\
\hline
$p_s^{\rm MS}$  & The price set by seller $s$ under the multi-seller strategy \\
\hline
$\S_C$  & The set of sellers who are able to achieve a positive profit under the multi-seller strategy \\
\hline
$\pi_s(\p[t])$  & The profit of seller $s$ in time slot $t$  \\
\hline
$\delta_s$  & The discount factor of seller $s$  \\
\hline
$h[t]$  & The price profile history till time slot $t$  \\
\hline
$\Pi_s(h[\infty])$  & The long-term discounted total profit of seller $s$  \\
\hline
$P_k$  & The probability of $k$ buyers arriving at the market in a time slot  \\
\hline
$b_s$  & The capacity of seller $s\in \S$ in Sections \ref{sec:limitedone} and \ref{sec:limitedmulti} \\
\hline
 \end{tabular}
\label{table}
\end{table}

\subsubsection{Sellers' Pricing Decisions in Stage I}\label{sec:seller}

Sellers compete to attract buyers in Stage I, and we can model their interactions as a game. 
\begin{game}[Single-Slot Static Seller Competition Game with Unlimited Capacities]\label{SCgame}
$\mbox{}$
\begin{itemize}
\item Players: the set $ \S $ of sellers. 
\item Strategies: each seller $s\in\mathcal{S}$ chooses a price $p_s \in [c,\rho r_s]$. 
\item Payoffs: each seller $s\in\mathcal{S}$ obtains a profit $\pi_s(p_s,p_{-s})=(p_s-c)q_{s}^\ast(p_s,p_{-s})$.
\end{itemize}
\end{game}
Here $c$ is sellers' homogeneous cost for producing one unit of product. 
The price upper bound $\rho r_s$ is due to the nonnegativity requirement of the optimal consumption $\rho/p_s-1/r_s$ in \eqref{eq:buyerOpt}.

Now we analyze the Nash equilibrium (NE) of the single-slot Stage I game, i.e., Game \ref{SCgame}. 
The NE has different forms depending on the number of sellers achieving the highest reputation. 
Recall in Lemma \ref{lemma:sellerSelection} that only the seller with the highest reputation-price ratio can get a positive demand. 
We let $r_{\rm{max}}$ be the highest reputation in the network, i.e., $r_{\rm{max}} = \max_{i \in \S} r_i$, and let $s^{\max}$ denote the smallest index of such sellers. 
We let $r_{\rm{sec}}$ be the second highest level of reputation, i.e., $r_{\rm{sec}} = \max_{i \in \S \setminus \{s^{\rm{max}}\}} r_i$. 
If multiple sellers have the same highest level of reputation, then $r_{\rm{sec}}=r_{\rm{max}}$. 
As an example, consider four sellers with reputations equal to $\{4, 9 , 1, 9\}$. 
Then $r_{\rm{max}} = 9, s^{\max} = 2$, and $r_{\rm{sec}} = 9$.

\emph{Case I:} Multiple sellers achieve the same highest reputation, i.e., $r_{\rm{sec}}=r_{\rm{max}}$. 
\begin{lemma}\label{lemma:NES1}
If $r_{\rm{sec}}=r_{\rm{max}}$, the unique Nash equilibrium of Game \ref{SCgame} is a price profile $\boldsymbol{p}^\ast$ such that 
\begin{equation*}
p_s^{\ast}=c,~\forall s \in\S.
\end{equation*}
\end{lemma}

\begin{proof}
See Appendix \ref{app:lemma_NES1}. % in the supplementary material.
\end{proof}

Lemma \ref{lemma:NES1} shows that a fierce market competition forces all sellers to set prices equal to the marginal price $c$ and obtain zero profits.

\emph{Case II:} Only one seller achieves the highest reputation, i.e., $r_{\rm{sec}}<r_{\rm{max}}$. 
In this case, the unique highest reputation seller achieves a positive profit, while all other sellers achieve zero profits. 
We let $\varepsilon$ denote a sufficiently small positive number.\footnote{In our paper, we assume that there is a minimum increment in price, and we set $\varepsilon=1$ cent.} 

\begin{lemma}\label{lemma:NES2}
If $r_{\rm{sec}}<r_{\rm{max}}$, the unique Nash equilibrium of Game \ref{SCgame} is a price profile $\boldsymbol{p}^\ast$ such that 
\begin{equation*}%\label{eq:NES}
p_{s}^\ast=
\left\{
\begin{aligned}
&\min \left\{ \frac{r_{\rm{max}}}{r_{\rm{sec}}}c-\varepsilon, \sqrt{c \rho r_s} \right\} ,\mbox{ if } r_s = r_{\rm{max}}; 
\\ 
&~c, \mbox{ otherwise}.
\end{aligned}
\right.
\end{equation*}
Here $\sqrt{c \rho r_s}$ is seller $s$'s monopoly price, i.e., the price that seller $s$ would choose to maximize his profit if he is the only seller in the market.
\end{lemma}

\begin{proof}
See Appendix \ref{app:lemma_NES2}. % in the supplementary material.
\end{proof}

The quantity $\frac{r_{\rm{max}}}{r_{\rm{sec}}}c-\varepsilon$ in Lemma \ref{lemma:NES2} enables seller $s$ with $r_s=r_{\max}$ to set the reputation-price ratio $r_s/p_s$ slightly higher than all other sellers' reputation-price ratios, and hence attracts all buyers. 
Furthermore, we have the following result.

\begin{corollary}\label{coro:monopoly}
If $r_{\rm{max}}> {\rho (r_{\rm{sec}})^2}/{c}$, then the highest reputation seller will set his monopoly price at the equilibrium, i.e., $$p_s^{\ast}=\sqrt{c \rho r_s}, \mbox{ if } r_s=r_{\max}, $$ and achieve a profit equal to his maximum profit achieved in a monopoly market. 
\end{corollary}

Theorem \ref{theo:NES} summarizes the results of Lemmas \ref{lemma:NES1} and \ref{lemma:NES2} with a unified expression. 

\begin{theorem}\label{theo:NES}
The unique Nash equilibrium of Game \ref{SCgame} is a price profile $\boldsymbol{p}^{\ast \rm NE-U}=(p_s^{\ast \rm NE-U}:\forall s\in\S)$ such that\footnote{The superscript ``U'' here represents unlimited capacity.} 
\begin{equation}\label{eq:pNES}
p_{s}^{\ast \rm NE-U}=
\left\{
\begin{aligned}
&\max \left\{ \min \left\{ \frac{r_{\rm{max}}}{r_{\rm{sec}}}c-\varepsilon, \sqrt{c \rho r_s} \right\} , c \right\}, \\
&~ \quad \mbox{ if } r_s = r_{\rm{max}};
\\
&~c, \mbox{ otherwise}.
\end{aligned}
\right.
\end{equation} 
\end{theorem}

\subsection{Dynamic Game Analysis}\label{subsec:DGUC}

Now we analyze the more realistic infinite-horizon dynamic game, where sellers need to make pricing decisions at the beginning of each infinitely many time slots repeatedly. 
As explained in Section \ref{sec:model}, it is reasonable to assume that buyers are bounded rational in such repeated interactions \cite{boundedrationality}, hence each buyer is myopic and chooses to maximize his payoff only in the current time slot as in Section \ref{sec:unlimitedSingle} \cite{shortbuyer}. 
The difference between the dynamic game here and the static single-slot game in Section \ref{sec:unlimitedSingle} lies in the sellers' decisions.

In a dynamic game, each seller chooses the prices over time to maximize his long-term discounted total profit. 
We denote the sellers' price profiles in time slot $t$ as $\p[t]$. 
We let $\pi_s(\p[t],\r[t])$ denote the profit of seller $s$ (for each $s\in\S$) in time slot $t$, which depends on all buyers' seller selection decisions and hence depends on all sellers' reputations $\r[t]=\{r_s[t]:\forall s\in\S\}$ in time slot $t$. 
As a result, each seller's long-term discounted total profit depends on all sellers' reputations which evolve dynamically over time. 

The reputation evolution process of seller $s\in\S$ can be described as follows: 
\begin{equation}\label{eq:repdyn}
r_s[t+1]=\frac{\omega_s X_s[t+1]}{ \sum_{i\in\S} \omega_i X_i[t+1]}=\frac{\omega_s \left(X_s[t]+q_s[t]\right)}{ \sum_{i\in\S} \omega_i \left(X_i[t]+q_i[t]\right)}.  
\end{equation}
Here $q_s[t]$ is the number of completed transactions of seller $s$ in the single time slot $t$, which depends on sellers' competition in time slot $t$ and hence depends on $\p[t]$ and $\r[t]$. 
As mentioned before, we assume that a buyer's average rating $\omega_s$ changes much slower than the change in the number of transactions and hence is assumed to be fixed in our current study. 
However, from a computational point of view, it is quite complex for each seller to compute the reputation evolution dynamics given sellers' current strategy choices in the market \cite{MyopicLearning}. 
And from a practical point of view, it is unlikely that each seller would explicitly compute his competitors' future best response pricing strategies and reputations, which is an implausible task in an actual game due to bounded rationality (e.g., bounded computational power) \cite{MyopicLearning}. 

Motivated by the above discussions, we consider a computationally simpler and more natural scenario, where sellers are myopic and play the dynamic game by assuming that all sellers' future reputations remain the same as the current reputations \cite{MyopicLearning}. 
This assumption requires only a weak form of rationality from the sellers, which has some similarities with the models of predictive control and receding horizon control \cite{ModelPreControl}, both of which are popular approaches to complex dynamic control problems. 
Such an assumption is plausible especially in a market with a large number of sellers. 
In such a market, sellers' myopic behavior is computationally simple; by contrast, it is an unreasonable computational requirement to sellers to solve a dynamic program with full knowledge (or accurate prediction) of future reputations of other sellers in every time slot.

Specifically, at the beginning of each time slot $t$, sellers observe the current reputation profile $\r[t]$, and conjecture that the reputation profile will \emph{remain constant for all time}; with this conjecture, each seller computes an optimal strategy over time (from time slot $t$ to the infinite future) to maximize his long-term discounted profit, and chooses the price in time slot $t$ accordingly. 
In the next time slot $t+1$, all sellers' reputations will evolve according to \eqref{eq:repdyn}. 
Then each seller repeats the same decision process as in time slot $t$, i.e., computing a new optimal strategy for time $t$ and future time slots based on the newly observed reputation profile $\r[t+1]$, and then implement the corresponding price in time slot $t+1$.

Now we explain how each myopic seller computes his long-term profit.
We let $h[t]$ denote the price history up to time slot $t$, i.e., all price profiles during the previous time slots, 
\begin{equation*}
h[t] \eq [\p[0],\p[1],\ldots, \p[t-1]]. 
\end{equation*}
Then, the conjectured long-term discounted total profit of seller $s$ in time slot $t$ is
\begin{equation}\label{eq:RGpi}
\Pi_s(h[t],\r[t]) \eq \sum_{t'=t}^\infty \delta_s^{t'-t} \pi_s(\p[t'],\r[t]).
\end{equation}
Here $\delta_s\in [0,1)$ is the time discount factor of seller $s$ \cite{Srikant}, and $\pi_s(\p[t'],\r[t])$ is the profit of seller $s$ in time slot $t'$ assuming that the sellers' reputation profile remains the same as $\r[t]$ in time slot $t$.

We model the sellers' infinitely repeated competition as a dynamic game as follows. 
\begin{game}[Dynamic Seller Competition Game with Unlimited Capacity]\label{DSCuc}
$\mbox{}$
\begin{itemize}
\item Players: the set $ \S $ of sellers. 
\item Strategies: each seller $s\in\mathcal{S}$ chooses the price $p_s[t]$ in each time slot $t\in[0,...\infty)$. 
\item Histories: the price profile history $h[t]$ till time slot $t$, for each $t\in[0,...\infty)$. 
\item Payoffs: the conjectured long-term discounted total profit of each seller $\Pi_s(h[t],\r[t]), \forall s \in \S$, for each $t\in[0,...\infty)$. 
\end{itemize}
\end{game}

Next we characterize the subgame perfect Nash equilibrium (SPNE) of Game \ref{DSCuc}. 
According to the Folk Theorem \cite{Srikant}, any feasible and individually rational strategy can become an equilibrium in the infinitely repeated game under proper discount factors $\delta_s, \forall s\in\S$. 
Depending on the highest reputation seller's minmax profit in the single-slot game, i.e., the profit under NE $\p^{\ast \rm NE-U}$ in Theorem \ref{theo:NES}, we discuss two types of SPNEs of the dynamic game which correspond to two different cases in the dynamic game: (i) a monopoly market where the seller with the highest reputation dominates the market, and (ii) a multi-seller market where a subset of sellers with relatively high reputations survive at the SPNE.

\subsubsection{Monopoly Market}

As described in Theorem \ref{theo:NES} for the single-slot game, only the seller with the highest reputation can win the seller competition game and achieve a positive profit. 
Furthermore, if the highest reputation seller sets his monopoly price at the NE, he can achieve a profit equal to his maximum profit achieved in a monopoly market. 
Now we turn to the dynamic game case. 
We will show that if the highest reputation seller is able to choose his monopoly price at the single-slot game NE, he will choose the same monopoly price at the SPNE of the dynamic game in the infinite time horizon.

\begin{theorem}\label{theo:NErNE}
If in time slot $t$,
\begin{equation}\label{eq:rmrn}
r_{\rm{max}}[t] > \frac{\rho \left(r_{\rm{sec}}[t]\right)^2}{c},
\end{equation}
then the unique SPNE of Game \ref{DSCuc} is that all sellers choose the price profile according to the NE of the single-slot game (as in Theorem \ref{theo:NES}) starting from time slot $t$, regardless of the history $h[t]$. 
\end{theorem}

\begin{proof}
See Appendix \ref{app:theo_NErNE}. % in the supplementary material.
\end{proof}

Theorem \ref{theo:NErNE} describes a monopoly market where the seller with the highest reputation dominates the market.

Based on our previous discussions, we know that as the seller reputations evolve over time, the bounded rational sellers might derive different SPNEs over time. 
However, we can show that as long as condition \eqref{eq:rmrn} is satisfied in a time slot $t$, it will always be satisfied in all later time slots. 
This means that Theorem \ref{theo:NErNE} will remain true for all time slots $t'>t$. 
Meanwhile, the reputation of all sellers will evolve as follows:

\begin{theorem}\label{theo:RepEvoMono}
If condition \eqref{eq:rmrn} is satisfied in time slot $t$, the reputation evolution process described in \eqref{eq:qsNES}, \eqref{eq:pNES}, and \eqref{eq:repdyn} starting from time slot $t$ converges to a stable state where 
\begin{equation*}
\begin{aligned}
& r_{s^{\max}}[\infty]=1,\\
& r_{s}[\infty]=0, \forall s\in\S,s\neq s^{\max}.
\end{aligned}
\end{equation*}
\end{theorem}

\begin{proof}
See Appendix \ref{app:theo_RepEvoMono}. % in the supplementary material.
\end{proof}

Intuitively, when the seller with the highest reputation dominates the market starting from time slot $t$, it will dominate the market and get a positive demand from all buyers in all later time slots. 
Hence, as time evolves, the highest reputation seller's reputation keeps monotonically increasing, and all other sellers' reputations keep monotonically decreasing. 
Finally, the highest reputation seller's reputation goes to 1 while all other sellers' reputations go to 0.

\subsubsection{Multi-seller Market}

Next we analyze the scenario when the highest reputation seller in the dynamic game can achieve a higher profit by not playing the NE of the single-slot game. 
We will first discuss the possibility of a profit maximum, i.e., a strategy whereby sellers maximize their joint profits in each time slot. 
Although the profit maximum is often NOT an equilibrium of the single-slot game, it can be enforced as part of the SPNE in the dynamic game.

Now we derive the profit maximum strategy for sellers in each time slot $t$. 
We let $s^{\max}[t]$ denote the smallest index of highest reputation sellers in time slot $t$. 
We let superscript ``NE-U'' denote the single-slot Nash equilibrium strategy in Theorem \ref{theo:NES}, and let superscript ``MS-U'' denote a price strategy in a multi-seller market in the unlimited capacity scenario defined in Lemma \ref{lemma:coop}. 
We let $S_r[t]$ denote the number of sellers with the highest reputation-price ratio, and let $\pi_{s^{\max}[t]}$ denote the single-slot profit of the highest reputation seller $s^{\max}[t]$ achieved from each buyer.

\begin{lemma}\label{lemma:coop}
Consider a price profile $\boldsymbol{p}^{\rm{MS-U}}[t]$ as follows: 
\begin{equation}\label{eq:pcoopun}
p_s^{\rm{MS-U}}[t]=\max \left\{ \frac{r_s[t]}{\varphi^{\rm{MS-U}}[t]} ,c \right \}, \forall s \in \S ,
\end{equation}
where
\begin{equation}\label{eq:phicoop}
\varphi^{\rm{MS-U}}[t] = \sqrt{\frac{r_{\rm{max}}[t]}{c \rho}}.
\end{equation}
Such a price vector $\boldsymbol{p}^{\rm{MS-U}}[t]$ is the unique pricing strategy that maximizes all sellers' total profits in a single time slot if 
\begin{equation}\label{eq:pim}
\frac{\pi_{s^{\max}[t]}^{\rm MS-U}[t]}{S_r^{\rm MS-U}[t]} \geq \frac{\pi_{s^{\max}[t]}^{\rm NE-U}[t]}{S_r^{\rm NE-U}[t]}.
\end{equation}
\end{lemma}

\begin{proof}
See Appendix \ref{app:lemma_coop}. % in the supplementary material.
\end{proof}

%\emph{Proof:} See Appendix \ref{app:lemma_coop} in the supplementary material. $\hfill \Box$

Condition \eqref{eq:pim} indicates that the highest reputation seller can achieve a higher or equal profit under the multi-seller strategy than that under the NE of the single-slot game.

The multi-seller price vector $\boldsymbol{p}^{\rm{MS-U}}[t]$ defined in \eqref{eq:pcoopun} divides the sellers into two groups: \emph{surviving sellers} who set the prices higher than $c$, i.e., $p_s^{\rm{MS-U}}[t]={r_s[t]}/{\varphi^{\rm{MS-U}}[t]}>c$, and \emph{non-surviving sellers} who set the prices equal to the marginal cost, i.e., $p_s^{\rm{MS-U}}[t]=c$. 
Under the multi-seller strategy, all surviving sellers have the same reputation-price ratio $\varphi^{\rm{MS-U}}[t]$ defined in \eqref{eq:phicoop} and can achieve a positive profit since $p_s^{\rm{MS-U}}[t]>c$. 
In the following, we derive the condition under which a seller can achieve a positive profit under the multi-seller strategy.

\begin{corollary}\label{coro:rth}
A seller $s\in \S$ achieves a positive single-slot profit under the multi-seller price profile $\boldsymbol{p}^{\rm{MS-U}}[t]$ if and only if 
\begin{equation}\label{eq:rth}
r_s[t] > r_{\rm{th}}[t] \eq \sqrt{\frac{c r_{\rm{max}}[t]}{\rho}}.
\end{equation}
\end{corollary}

\begin{proof}
See Appendix \ref{app:coro_rth}. % in the supplementary material.
\end{proof}

Corollary \ref{coro:rth} implies that a seller with a low reputation will be driven out of the market in such a multi-seller market. 
We denote the set of surviving sellers who satisfy $r_s[t] > r_{\rm{th}}[t]$ by $\S_C[t]$ with a size of $S_C[t]$.

Although in general the multi-seller price profile is not an NE in the single-slot game (as a surviving seller has an incentive to decrease the price to undercut other sellers), we show that such a multi-seller price strategy can be enforced as an SPNE by a punishment strategy in the infinite-horizon dynamic game.\footnote{The dynamic seller competition game is a standard non-cooperative game where sellers are competing with each other freely and fairly. There may be multiple SPNEs, and the one in Theorem \ref{theo:NEcoop} is one SPNE under which more than one seller achieves a positive profit by setting different prices without enforcing any kind of coordination agreement. If any seller deviates from the strategy, other sellers will play the price war strategy in Theorem \ref{theo:NES} as a punishment. Hence the deviating seller will not get better off.} 
One effective punishment strategy is the Friedman punishment, where sellers revert to the NE $\p^{\ast \rm NE-U}$ in Theorem \ref{theo:NES} if anyone deviates from multi-seller price strategy \cite{Srikant}, to be explained next.\footnote{It is feasible for each seller to monitor others' public prices and reputations. Furthermore, computing the multi-seller price profile according to Lemma \ref{lemma:coop} is computationally easy.} 
We let $S_L[t]$ denote the number of sellers with the highest reputation in time slot $t$.

\begin{theorem}\label{theo:NEcoop}
Consider the following strategy profile: all sellers set the multi-seller price profile $\boldsymbol{p}^{\ast \rm SPNE-U}[t]=\boldsymbol{p}^{\rm{MS-U}}[t]$ in Lemma \ref{lemma:coop} in each time slot $t$ until a seller deviates, in which case all sellers choose the price profile according to the NE $\p^{\ast \rm NE-U}$ in Theorem \ref{theo:NES} in all future time slots. 
Such a strategy profile is an SPNE if 
\begin{align}
&\frac{\pi_{s^{\max}[t]}^{\rm{MS-U}}[t]}{S_C[t]} \geq \frac{\pi_{s^{\max}[t]}^{\rm NE-U}[t]}{S_L[t]},\label{con:pim} \mbox{ and }\\
&\delta_s > \frac{{S}_C[t]-1}{{S}_C[t]} \cdot \frac{{\pi}_s^{\rm{MS-U}}[t]}{{\pi}_s^{\rm{MS-U}}[t] - {\pi}_s^{\rm NE-U}[t]/{S}_L[t]}, \forall s \in S.\label{con:delta}
\end{align}
\end{theorem}

\begin{proof}
See Appendix \ref{app:theo_NEcoop}. % in the supplementary material.
\end{proof}

Condition \eqref{con:pim} indicates that the highest reputation seller can achieve a higher or equal profit under the multi-seller strategy than that under the NE of the single-slot game, and hence is willing to play the multi-seller strategy at SPNE. 
Condition \eqref{con:delta} indicates that sellers are sufficiently patient, and hence can achieve higher profits by playing the multi-seller strategy in the long run at the SPNE than deviating from it.

\section{Limited Capacity and One Buyer Per Seller Scenario}\label{sec:limitedone}

In this section, we look at a different scenario where each seller has a limited capacity and can only serve one buyer in each time slot. 
This scenario is motivated by practical online sharing economy platforms such as Airbnb, where a host of one apartment typically only serves one buyer (corresponding to one traveling group) at any given time. 
Let $b_s$ denote the capacity of seller $s$. 
The limit of serving one buyer does not imply that $b_s=1$, as a single buyer can request multiple products (e.g., one family using two rooms in an apartment). 
Later in Section \ref{sec:limitedmulti}, we further consider the scenario where each seller still has limited capacity but can serve multiple buyers in each time slot.

Similar to Section \ref{sec:unlimited}, we first analyze the single-slot game and then analyze the dynamic game.

\subsection{Single-Slot Game Analysis}\label{sec:limitedSingle}

Recall that the single-slot game as shown in Figure 2(b) is a Stackelberg game.   
We first analyze the buyers' purchasing decision in Stage II, and then study the sellers' pricing decisions in Stage I. 
Since we focus on each single time slot, we will suppress the time index $t$ in Section \ref{sec:limitedSingle}. 
We will bring the index $t$ back in Section \ref{subsec:DGLC}.

\subsubsection{Buyers' Purchasing Decisions in Stage II}

In line with the analysis in Section \ref{sec:unlimitedSingle}, each buyer prefers to choose the seller with the highest reputation-price ratio. 
However, the conclusion in Section \ref{sec:unlimitedSingle} is not directly applicable here, as each seller can only serve one buyer in each time slot. 
Once a seller is chosen by a buyer, other buyers cannot choose the same seller in the same time slot.

\subsubsection{Sellers' Pricing Decisions in Stage I}

In the following, we first define the seller competition game, and then analyze the NE considering a simple case with a small number of sellers. 
We finally analyze the NE for the general case.

To describe the seller competition game, we first calculate a seller's profit which is a product term of the probability of the seller being chosen by a buyer and the consumption amount of the corresponding buyer. 
Since each seller has a limited capacity and can only serve one buyer in each time slot, a seller's profit depends on not only his reputation-price ratio but also the buyer arrival process. 
This is the key difference between the analysis here and the analysis of the unlimited capacity case in Section \ref{sec:unlimited}.

Recall that we denote the probability of having $k$ buyers arriving in a time slot by $P_k$. 
When this happens, a seller within the highest $k$ reputation-price ratios will be able to serve a buyer. 
We let $g_s^{\rm th}(p_s,p_{-s})$ denote the rank of the reputation-price ratio of seller $s\in\S$, i.e., seller $s$ has the $g_s^{\rm th}(p_s,p_{-s})$-th highest reputation-price ratio and he can serve a buyer if no fewer than $g_s^{\rm th}(p_s,p_{-s})$ buyers arrive in the market in a time slot. 
For example, if seller $s$ has the $4$-th highest reputation-price ratio, then $g_s^{\rm th}(p_s,p_{-s})=4$.

Now we calculate the probability of seller $s$ being chosen by a buyer.
We define the conditional probability $E_s(p_s,p_{-s}|k)$ as the probability of seller $s$ being chosen by a buyer if the total number of buyers arriving in the time slot is $k$, which can be calculated as:
\begin{equation}
E_s(p_s,p_{-s}|k)=
\left\{
\begin{aligned}
& 1, \mbox{ if } k \geq g_s^{\rm th}(p_s,p_{-s});
\\%[-2pt]
& 0, \mbox{ otherwise}.
\end{aligned}
\right.
\end{equation}
We let $\mathbb{E}_s(p_s,p_{-s})$ denote the (unconditional) probability of seller $s$ being chosen by a buyer given buyers' random arrivals. 
The value of $\mathbb{E}_s(p_s,p_{-s})$ depends on the prices and the buyer arrival distribution: 
\begin{equation}
\mathbb{E}_s(p_s,p_{-s}) = 1- \sum_{k=0}^{g_s^{\rm th}(p_s,p_{-s})-1}P_k.
\end{equation}

If seller $s$ is chosen by a buyer, the buyer's optimal consumption amount (as calculated in Lemma \ref{lemma:Optx1}) $x_{s}^\ast=\rho/p_s-1/r_{s}$ reaches the seller's capacity $b_s$ when $p_s \leq \rho/(b_s+1/r_s)$. 
In this case, the profit of seller $s$ obtained from the buyer is $(p_s-c)b_s$, which increases with the price $p_s$. 
Hence it is never optimal for the seller to choose a price that is strictly smaller than $\rho/(b_s+1/r_s)$. 
Combining the fact that the price $p_s$ should be no smaller than the marginal cost $c$, we have the following lower bound for $p_s$:
\begin{equation}\label{eq:psmin}
p_s \geq p_s^{\min} \eq \max \left\{ c,\frac{\rho}{b_s+1/r_s} \right\}.
\end{equation}
Hence, the profit of each seller $s\in\S$ is 
\begin{equation}\label{eq:piLO}
\pi_s(p_s,p_{-s})=(p_s-c)\left(\frac{\rho}{p_s}-\frac{1}{r_s}\right)\mathbb{E}_s(p_s,p_{-s}), 
\end{equation}
for $p_s \in \left[ p_s^{\min},\rho r_s \right]$.

Sellers compete to attract buyers, and we can model their interactions as a game.
\begin{game}[Single-Slot Static Seller Competition Game with Limited Capacity and One Buyer per Seller]\label{SCgameCLRA}
\emph{
\begin{itemize}
\item Players: the set $ \S $ of sellers. 
\item Strategies: each seller $s\in\mathcal{S}$ chooses a price $p_s \in \left[ p_s^{\min},\rho r_s \right]$, where $p_s^{\min}$ is defined in \eqref{eq:psmin}. 
\item Payoffs: each seller $s\in\mathcal{S}$ obtains a profit $\pi_s(p_s,p_{-s})$ in \eqref{eq:piLO}. 
\end{itemize}
}
\end{game}

Now we analyze the NE of Game \ref{SCgameCLRA}, starting from the simple example of three sellers. 
Then we will present the analysis for more general cases.

In the three-seller example, without loss of generality, we assume that $r_1 \geq r_2 \geq r_3$. 
The NE of the three-seller example can be described in three cases, depending on whether a lower reputation seller has an incentive to set a low price to compete with a higher reputation seller. 
Intuitively, all three sellers will set their monopoly prices at NE when the reputation gap between each pair of two adjacent sellers is relatively large, and hence lower reputation sellers cannot benefit by reducing their prices from the monopoly prices. 
However, if the reputation gap between two adjacent sellers is small, then the analysis will be more complicated.

We let $\pi_s^{\rm mon}$ denote the profit of seller $s\in\S$ when all sellers choose their monopoly prices. 
We let $\pi_s^{\rm dev}$ denote the profit of seller $s$ achieved by setting a low price $p_s^{\rm dev}=\frac{r_s}{r_{s-1}}p_{s-1}^{\rm mon}-\varepsilon$ to compete with seller $s-1\in\S$ when all other sellers choose their monopoly prices.

\begin{lemma}\label{lemma:3seller}
The unique Nash equilibrium of three-seller Game \ref{SCgameCLRA} where $r_1 \geq r_2 \geq r_3$ is a price profile $\boldsymbol{p}^{\ast \rm NE-LO}$ such that:\footnote{The superscript ``LO'' represents limited capacity and one buyer per seller.}
\begin{itemize}
\item Case I: 
If
\begin{equation}\label{eq:3sellercon1}
\pi_2^{\rm mon} \geq \pi_2^{\rm dev}, \pi_3^{\rm mon} \geq \pi_3^{\rm dev},
\end{equation}
then
\begin{equation}\label{eq:3sellerNE1}
\begin{aligned}
&p_1^{\ast \rm NE-LO} = \sqrt{c \rho r_1}, \\
&p_2^{\ast \rm NE-LO} = \sqrt{c \rho r_2}, \\
&p_3^{\ast \rm NE-LO} = \sqrt{c \rho r_3}.
\end{aligned}
\end{equation}
\item Case II: 
If
\begin{equation}\label{eq:3sellercon3}
\pi_2^{\rm mon} < \pi_2^{\rm dev}, 
\pi_3^{\rm mon} \geq \pi_3^{\rm dev},
\end{equation}
then
\begin{equation}\label{eq:3sellerNE2}
\begin{aligned}
&p_1^{\ast \rm NE-LO} = \max \left \{  \min \left\{ \sqrt{c \rho r_1}, \frac{r_1}{r_2}p_2^{\ast \rm NE-LO}-\varepsilon \right\}, p_1^{\min}  \right \},\\
&  p_2^{\ast \rm NE-LO} = \max \left\{ \min \left\{ \sqrt{c \rho r_2}, \frac{r_2}{r_3}p_3^{\ast \rm NE-LO}-\varepsilon \right\}, p_2^{\min} \right\}, \\
&  p_3^{\ast \rm NE-LO} = \sqrt{c \rho r_3}.
\end{aligned}
\end{equation}
\item Case III: 
If
\begin{equation}\label{eq:3sellercon4}
\pi_3^{\rm mon} < \pi_3^{\rm dev},
\end{equation}
then
\begin{equation}\label{eq:3sellerNE3}
\begin{aligned}
& p_1^{\ast \rm NE-LO} = \max \left\{  \min \left\{ \sqrt{c \rho r_1}, \frac{r_1}{r_2}p_2^{\ast \rm NE-LO}-\varepsilon \right\}, p_1^{\min}  \right\},  \\
&p_2^{\ast \rm NE-LO} = \max \left \{  \min \left\{ \sqrt{c \rho r_2}, \frac{r_2}{r_3}p_3^{\ast \rm NE-LO}-\varepsilon \right\},  p_2^{\min}  \right\}, \\
& p_3^{\ast \rm NE-LO} = p_3^{\min}.
\end{aligned}
\end{equation}
\end{itemize}
\end{lemma}

\begin{proof}
See Appendix \ref{app:lemma_3seller}. % in the supplementary material.
\end{proof}

In Case I, condition \eqref{eq:3sellercon1} implies that seller $2$ and seller $3$ cannot improve their profits by lowering their prices. 
Hence \eqref{eq:3sellerNE1} indicates that each seller sets his monopoly price at the NE. 
In Case II, condition \eqref{eq:3sellercon3} indicates that seller $2$ wants to lower his price while seller $3$ does not. 
Hence at the NE as in \eqref{eq:3sellerNE2}, seller $3$ still sets his monopoly price, while sellers $1$ and $2$ engage in the price competition. 
In Case III, condition \eqref{eq:3sellercon4} indicates that even the lowest reputation seller $3$ wants to lower his price. 
Hence at the NE as in \eqref{eq:3sellerNE3}, seller $3$ lowers his price to his minimum level, and sellers $1$ and $2$ engage in the price competition involving all three sellers.

We now analyze the general case where $S \geq 2$, motivated by our analysis of the $S=3$ case. 
Without loss of generality, we assume $r_1 \geq r_2 \geq \cdots \geq r_S$. 
We propose the Monopoly Sequential Adjusting Algorithm (Algorithm \ref{algo:SAA}) to derive the NE of Game \ref{SCgameCLRA}. 
Based on the analysis of the three-seller example, we know that the strategy that every seller sets the monopoly price, i.e., $p_s^\ast=\sqrt{c \rho r_s},~(\forall s\in\S)$, is the Nash equilibrium if the reputation gaps among adjacent indexed sellers are relatively large (Line 1 of Algorithm \ref{algo:SAA}). 
However, if the reputation gap is small for some adjacent sellers such that a seller $i~(i\in\S)$ can earn more profit by setting a low price (Lines 4, 5, 6, 7 of Algorithm \ref{algo:SAA}), then all sellers having equal or higher reputations (i.e., sellers $1,2,\ldots,i$) will need to adjust their prices accordingly and engage in the price competition to reach the competitive NE (Lines 8, 9, 10 of Algorithm \ref{algo:SAA}). 
Game \ref{SCgameCLRA} is a complete information game where each seller knows all sellers' reputations $\boldsymbol{r}$, capacities $\boldsymbol{b}$, and marginal cost $c$. 
Every seller can independently compute the equilibrium by running Algorithm \ref{algo:SAA} locally without additional information exchange.

\begin{theorem}\label{theo:RandCLONE}
The price profile derived by Algorithm \ref{algo:SAA} is the unique Nash equilibrium of Game \ref{SCgameCLRA}.
\end{theorem}

\begin{proof}
See Appendix \ref{app:theo_RandCLONE}. % in the supplementary material.
\end{proof}

\begin{algorithm}[t]
\LinesNumbered
\SetAlgoLined
\begin{small}
\KwIn{Repuation $\boldsymbol{r}$, capacity $\boldsymbol{b}$, marginal cost $c$, product evaluation $\rho$, and a small positive $\varepsilon$}
\KwOut{Nash equilibrium pricing strategy $\boldsymbol{p}^{\ast \rm NE-LO}$}
Initiate $p_s^{\ast}=\sqrt{c \rho r_s}, \forall s \in \S$\\
Set $i=S$\\
\While{$p_i^\ast=\sqrt{c \rho r_i}$}{
Set $p_i^{\rm dev}=\frac{r_i}{r_{i-1}}p_{i-1}^{\ast}-\varepsilon$\\
$\displaystyle \pi_i^{\rm mon}=(p_i^{\ast}-c)\left(\frac{\rho}{p_i^{\ast}}-\frac{1}{r_i}\right)\left(1- \sum_{k=0}^{i-1}P_k \right)$\\
$\displaystyle  \pi_i^{\rm dev}= \left(p_i^{\rm dev}-c \right)\left(\frac{\rho}{p_i^{\rm dev}}-\frac{1}{r_i}\right)\left(1- \sum_{k=0}^{i-2}P_k \right)$\\
\If{
$\pi_i^{\rm mon}<\pi_i^{\rm dev}$
}{
\For{$j=i:-1:1$}{
\eIf{$i=S$}{
Set $p_i^\ast=p_i^{\min}$\\
}{
Set $p_j^\ast=\max \left\{  \min\left\{ \sqrt{c \rho r_j}, \frac{r_j}{r_{j+1}}p_{j+1}^\ast-\varepsilon \right\}, p_j^{\min}  \right\}$\\
}
}
}
Set $i=i-1$\\
}
Set $\boldsymbol{p}^{\ast \rm NE-LO}=\boldsymbol{p}^{\ast}$\\
\end{small}
\caption{Monopoly Sequential Adjusting Algorithm}
\label{algo:SAA}
\end{algorithm}

\subsection{Dynamic Game Analysis}\label{subsec:DGLC}

Next we will discuss the infinite-horizon dynamic game parallel to Game \ref{SCgameCLRA}. 
Different from the unlimited capacity scenario, a monopoly market will not emerge in the limited capacity and one buyer per seller scenario, since each seller can only serve one buyer in each time slot and no seller can serve the demands from all buyers.
In the following, we discuss the SPNE of the dynamic game which corresponds to a multi-seller market. 
We first derive the multi-seller strategy that can maximize sellers' joint profits in a single slot, and then analyze the SPNE of the dynamic game where the multi-seller strategy is enforced by a punishment strategy. 

In the following lemma, we derive the profit maximum strategy for sellers in each time slot $t$.

\begin{lemma}\label{lemma:cooplimitedone}
There exists a unique multi-seller price profile $\boldsymbol{p}^{\rm{MS-LO}}[t]$ that maximizes all sellers' total profits in each time slot $t$: 
\begin{equation}\label{eq:pcoop}
p_s^{\rm{MS-LO}}[t]=\sqrt{c \rho r_s[t]}, \forall s \in \S.
\end{equation}
\end{lemma}

\begin{proof}
See Appendix \ref{app:lemma_cooplimitedone}. % in the supplementary material.
\end{proof}

Intuitively, with a limited supply and the capability of serving one buyer in each time slot, the strategy that each seller sets his monopoly price in every time slot is the unique strategy that maximizes sellers' joint profits. 
Different from the unlimited capacity scenario, here all sellers will set the multi-seller prices as their monopoly prices.

Although the multi-seller price profile may not be an NE in the single-slot game, we show that it can be enforced as an SPNE by a punishment strategy in the dynamic game. 
We let $\pi_s^{\rm{MS-LO}}[t]$ denote the expected profit of seller $s$ achieved under $\boldsymbol{p}^{\rm{MS-LO}}[t]$ in time slot $t$, let $\pi_s^{\rm NE-LO}[t]$ denote the profit of seller $s$ achieved at the NE of the single-slot game in time slot $t$ derived by Algorithm \ref{algo:SAA}, and let $\pi_s^{\rm{dev}}[t]$ denote the maximum expected profit that owner $s$ can achieve by unilaterally deviating from the multi-seller strategy in time slot $t$ (assuming that all other sellers choose according to the multi-seller price profile).

\begin{theorem}\label{theo:RandCLOSPNE}
Consider the following strategy profile: all sellers set the multi-seller price profile $\boldsymbol{p}^{\ast \rm SPNE-LO}[t]=\boldsymbol{p}^{\rm{MS-LO}}[t]$ in each time slot $t$ until a seller deviates, in which case all sellers choose the price profile according to the NE derived in Algorithm \ref{algo:SAA} in all future time slots. 
Such a strategy profile is an SPNE if 
\begin{equation}\label{con:deltalimitedone}
\delta_s > \frac{ {\pi}_s^{\rm{dev}}[t]-{\pi}_s^{\rm{MS-LO}}[t] }{ {\pi}_s^{\rm{dev}}[t]-{\pi}_s^{\rm{NE-LO}}[t]} , \forall s \in \S. 
\end{equation}
\end{theorem}

\begin{proof}
See Appendix \ref{app:theo_RandCLOSPNE}. % in the supplementary material.
\end{proof}

Different from the SPNE for the unlimited capacity scenario in Theorem \ref{theo:NEcoop}, here the highest reputation seller can always get a higher or equal profit under the multi-seller strategy than that under the NE of the single-slot game, and hence is always willing to play the multi-seller strategy.

\section{Limited Capacity and Multiple Buyers Per Seller Scenario}\label{sec:limitedmulti}

In this section, we analyze the scenario where each seller has a limited capacity and can serve multiple buyers in each time slot. 
This can well model small sellers on Amazon. 
We first analyze the single-slot game and then analyze the dynamic game.

\subsection{Single-Slot Game Analysis}\label{sec:LMSingle}

Recall that the single-slot game is a Stackelberg game.   
We first analyze the buyers' purchasing decision in Stage II, and then study the sellers' pricing decisions in Stage I. 
We will suppress the time index $t$ in the single-slot game analysis in Section \ref{sec:LMSingle}, and bring the index $t$ back in Section \ref{subsec:DGLM}.

\subsubsection{Buyers' Purchasing Decisions in Stage II}

As in Sections \ref{sec:unlimited} and \ref{sec:limitedone}, each buyer prefers to choose the seller with the highest reputation-price ratio. 
However, the conclusion in Section \ref{sec:limitedone} is not directly applicable here, as each seller can serve multiple buyers within his capacity in each time slot. 
Once a seller reaches his capacity, other buyers cannot further purchase from this seller in the same time slot.

\subsubsection{Sellers' Pricing Decisions in Stage I}

In the following, we first define the seller competition game, and then analyze the NE. 

To describe the seller competition game, we first calculate a seller's profit, which is a product of the consumption amount of a buyer at the seller and the expected number of buyers that the seller serves. 
If seller $s\in\S$ is chosen by a buyer, the buyer's optimal consumption amount (as calculated in Lemma \ref{lemma:Optx1}) is $x_s^\ast={\rho}/{p_s}-{1}/{r_s}$ as long as the price $p_s$ satisfies $p_s<\rho r_s$. 
For a small seller on Amazon, he cannot serve all buyers due to his limited capacity. 
%Mathematically, $n_s(p_s)$ calculated in \eqref{eq:nsLM} should satisfy $n_s(p_s)<\infty$. 
Equivalently, $x_s^\ast={\rho}/{p_s}-{1}/{r_s}>0$ and $x_s^\ast$ cannot be infinitesimal. 
We assume that when a buyer chooses a seller, the minimum amount of product that the buyer needs to buy is $x_0$ (e.g., one book $x_0=1$ on Amazon). 
This implies that if a buyer chooses seller $s$, the consumption amount should satisfy $x_s^\ast ={\rho}/{p_s}-{1}/{r_s} \geq x_0,$
which leads to an upper bound for $p_s$:
\begin{equation}\label{eq:psmax}
p_s \leq p_s^{\max} \eq \frac{\rho}{x_0+1/r_s}.
\end{equation}
The maximum number of buyers that seller $s$ with capacity $b_s$ can serve is 
\begin{equation}\label{eq:nsLM}
n_s(p_s)=\left \lfloor \frac{b_s}{{\rho}/{p_s}-{1}/{r_s}} \right \rfloor <\infty.
\end{equation}
Here the floor function $\left \lfloor x \right \rfloor$ denotes the minimum integer that is no larger than $x$. 
Note that $n_s(p_s)$ is non-decreasing in price $p_s$. 
Once seller $s$ serves $n_s(p_s)$ buyers, he cannot serve any other buyers.

Now we calculate the expected number of buyers that seller $s$ serves considering buyers' stochastic arrivals. 
Note that buyers will choose sellers in the decreasing order of the reputation-price ratio. 
Let $L_s(p_s,p_{-s})$ denote the number of buyers served by sellers whose reputation-price ratios are higher than that of seller $s$.
Sellers who have the same reputation-price ratio as seller $s$ will be chosen by buyers randomly with an equal probability. 
Let $\mathcal{C}_s$ denote the set of sellers who have the same reputation-price ratio as seller $s$, and hence $l_s(p_s,p_{-s})\eq\sum_{i\in\mathcal{C}_s}n_i(p_i)$ denotes the maximum number of buyers that sellers in set $\mathcal{C}_s$ can serve. 
We define $h_s^1=L_s(p_s,p_{-s})$ and $h_s^2=L_s(p_s,p_{-s})+l_s(p_s,p_{-s})$. 
Then we can calculate the expected number of buyers that seller $s$ serves, denoted by $\mathbb{N}_s(p_s,p_{-s})$, as follows:
\begin{equation}\label{eq:ENMB}
\begin{aligned}
&\displaystyle \mathbb{N}_s(p_s,p_{-s})=\max \bigg\{ 0, n_s(p_s) \cdot \min \bigg\{ 1, \\
& \displaystyle \sum_{k=h_s^1+1}^{h_s^2}\frac{k-h_s^1}{h_s^2-h_s^1}P_k+\sum_{k=h_s^2+1}^{\infty}P_k  \bigg\} \bigg\}.
\end{aligned}
\end{equation}
Hence, the expected profit of each seller $s\in\S$ is 
\begin{equation}\label{eq:piLM}
\pi_s(p_s,p_{-s})=(p_s-c)\left(\frac{\rho}{p_s}-\frac{1}{r_s}\right)\mathbb{N}_s(p_s,p_{-s}),
\end{equation}
for $p_s \in \left[ p_s^{\min},p_s^{\max} \right]$, where $p_s^{\min}$ and $p_s^{\max}$ are defined in \eqref{eq:psmin} and \eqref{eq:psmax}, respectively.

Sellers compete to attract buyers, and we can model their interactions as follows.
\begin{game}[Single-Slot Static Seller Competition Game with Limited Capacity and Multiple Buyers per Seller]\label{SCgameMB2}
\emph{
\begin{itemize}
\item Players: the set $ \S $ of sellers.
\item Strategies: each seller $s\in\mathcal{S}$ chooses a price $p_s \in \left[ p_s^{\min},p_s^{\max} \right]$, where $p_s^{\min}$ and $p_s^{\max}$ are defined in \eqref{eq:psmin} and \eqref{eq:psmax}, respectively. 
\item Payoffs: each seller $s\in\mathcal{S}$ obtains a profit $\pi_s(p_s,p_{-s})$ in \eqref{eq:piLM}.
\end{itemize}
}
\end{game}

We now analyze the NE for Game \ref{SCgameMB2}. 
Without loss of generality, we assume $r_1 \geq r_2 \geq \cdots \geq r_S$. 
We propose the Maximum Sequential Adjusting Algorithm (Algorithm \ref{algo:RSAA}) to derive the NE of Game \ref{SCgameMB2}. 
Different from Algorithm \ref{algo:SAA}, we initialize each seller's price as his maximum price (Line 1 of Algorithm \ref{algo:RSAA}), which will be the price that he will choose if he can sell his entire capacity. 
However, if a seller $i\in\S$ has an incentive to lower his price to gain more profit (Line 3 of Algorithm \ref{algo:RSAA}), then sellers $1,2,\ldots,i$ will compete and finally reach the Nash equilibrium (Lines 4, 5 of Algorithm \ref{algo:RSAA}).

\begin{algorithm}[t]
\LinesNumbered
\SetAlgoLined
\begin{small}
\KwIn{Repuation $\boldsymbol{r}$, capacity $\boldsymbol{b}$, marginal cost $c$, product evaluation $\rho$, and a small positive $\varepsilon$}
\KwOut{Nash equilibrium pricing strategy $\boldsymbol{p}^{\ast \rm NE-LM} $}
Initiate $p_s^\ast= p_s^{\max}, \forall s \in \S$ \\
\For{$i=S:-1:1$}{

\While{ $\pi_i(p_i^\ast - \varepsilon,p_{-i}^\ast) > \pi_i(p_i^\ast ,p_{-i}^\ast)$ }{
Set $\varphi=\frac{r_i}{p_i^\ast}+\epsilon$\\
Set $p_j^\ast=\frac{r_j}{\varphi}, \forall j=1,2,\ldots,i$\\
}

}
Set $\boldsymbol{p}^{\ast \rm NE-LM}=\boldsymbol{p}^{\ast}$
\end{small}
\caption{Maximum Sequential Adjusting Algorithm}
\label{algo:RSAA}
\end{algorithm}

\begin{theorem}\label{theo:RandCLM}
The price profile derived by Algorithm \ref{algo:RSAA} is the unique Nash equilibrium of Game \ref{SCgameMB2}.
\end{theorem}

\begin{proof}
See Appendix \ref{app:theo_RandCLM}. % in the supplementary material.
\end{proof}

When a seller can serve multiple buyers, he is able to sell out all his products and the price that maximizes his profit is his maximum price. 
However, when each seller can only serve one buyer, the optimal price is his monopoly price.

\begin{table*}[t]
%\vspace{-4mm}
\newcommand{\tabincell}[2]{\begin{tabular}{@{}#1@{}}#2\end{tabular}}
\centering
\caption{Comparison of Sections \ref{sec:unlimited}, \ref{sec:limitedone}, and \ref{sec:limitedmulti}}
\begin{tabular}{|c|c|c|c|}
% \hline
%   &    \textbf{Section \ref{sec:unlimited}} &    \textbf{Section \ref{sec:limitedone}} &    \textbf{Section \ref{sec:limitedmulti}} \\
\hline
\textbf{Scenario}  & \tabincell{c}{\textbf{Section \ref{sec:unlimited}}:\\ Unlimited  capacity}  & \tabincell{c}{\textbf{Section \ref{sec:limitedone}}:  Limited capacity \\and one buyer per seller} & \tabincell{c}{\textbf{Section \ref{sec:limitedmulti}}: Limited capacity \\and multiple buyers per seller}  \\
\hline
\tabincell{c}{\textbf{NE} for the \\ Single-Slot Game}  & \tabincell{c}{Only a seller with the \\ \emph{highest reputation} may obtain \\ a positive profit (Theorem \ref{theo:NES})}  & \tabincell{c}{Sellers with \emph{relatively high reputations} obtain \\ positive profits; with large reputation gaps,\\ they set their \emph{monopoly prices}, otherwise,\\ they set the competitive prices (Theorem \ref{theo:RandCLONE})} & \tabincell{c}{Sellers with \emph{relatively high reputations} obtain \\ positive profits; with large reputation gaps,\\ they set their \emph{maximum prices}, otherwise,\\ they set the competitive prices (Theorem \ref{theo:RandCLM})}  \\
\hline
\tabincell{c}{\textbf{SPNE} for the \\ Dynamic Game}  & \tabincell{c}{Both \emph{monoply market} and \\ \emph{multi-seller market} can exist \\ (Theorems \ref{theo:NErNE} and \ref{theo:NEcoop})}  & \tabincell{c}{Only \emph{multi-seller market} exists and \\ sellers may all set their \emph{monopoly} \\ \emph{prices} (Lemma \ref{lemma:cooplimitedone} and Theorem \ref{theo:RandCLOSPNE})} & \tabincell{c}{Only \emph{multi-seller market} exists and \\ sellers may all set their \\ \emph{maximum prices} (Lemma \ref{lemma:coopLM})}  \\
\hline
\end{tabular}
\label{table:comparison}
%\vspace{-3mm}
\end{table*}

\begin{figure*}[t]
%\vspace{-3mm}
 \centering
\begin{minipage}[t]{0.325 \linewidth}
\centering
\includegraphics[width=1\textwidth]{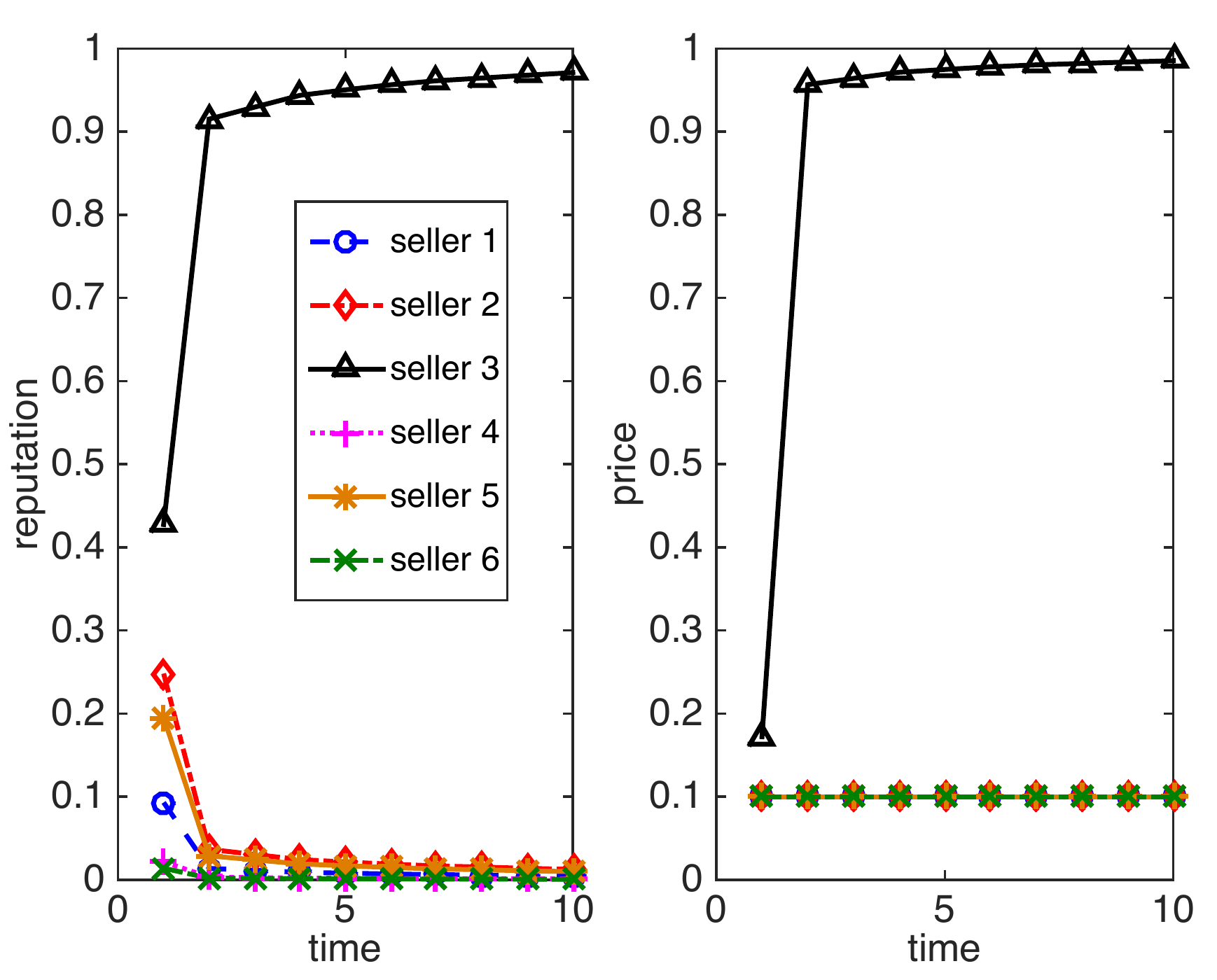}
\vspace{-3mm}
  \caption{Reputation and Price Dynamics (Unlimited Capacity, $\boldsymbol{X}$ and $\boldsymbol{\omega}^{I}$)}\label{fig:1Basic}
\end{minipage}
\begin{minipage}[t]{0.325 \linewidth}
\centering
\includegraphics[width=1\textwidth]{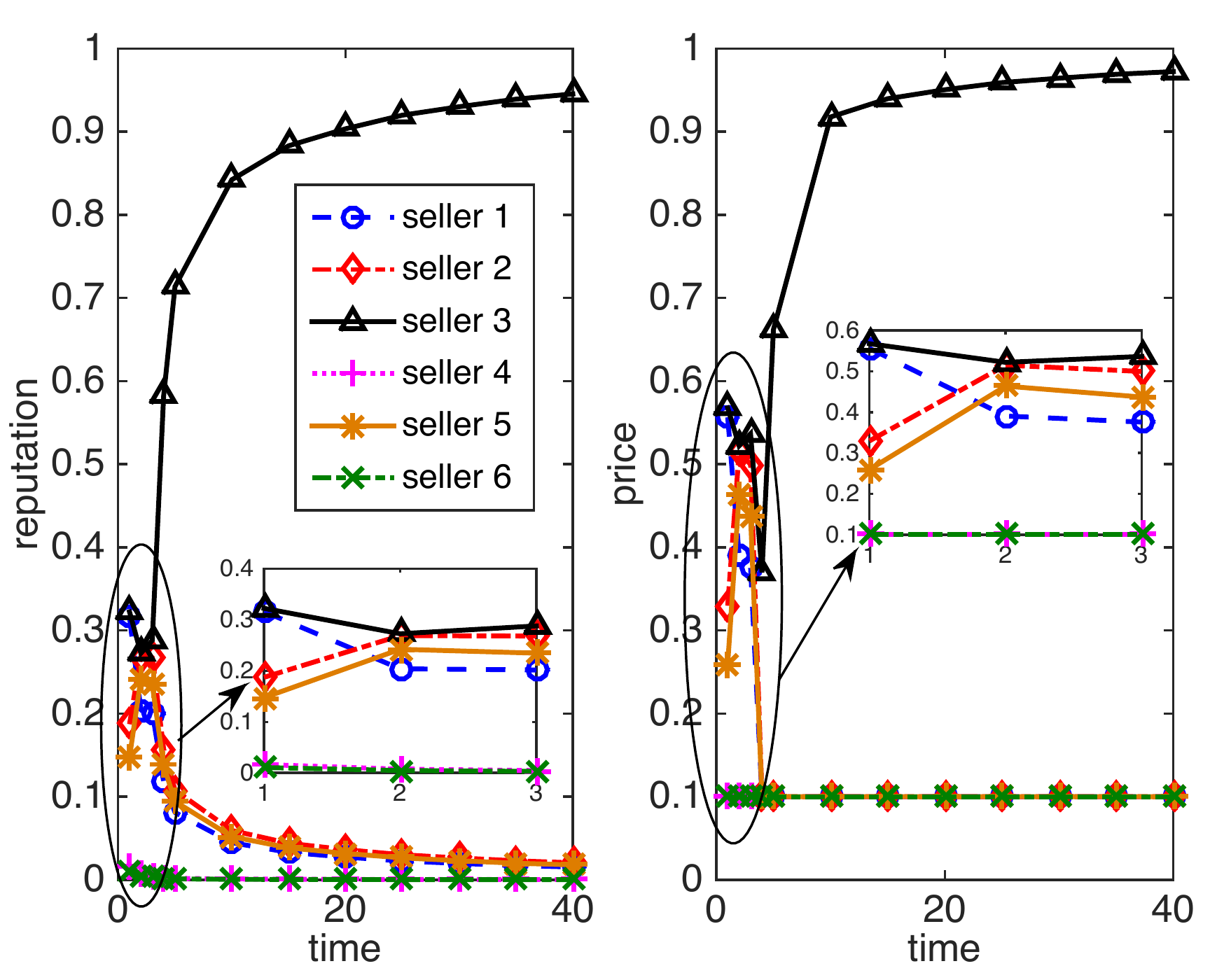}
\vspace{-3mm}
  \caption{Reputation and Price Dynamics (Unlimited Capacity, $\boldsymbol{X}$ and $\boldsymbol{\omega}^{II}$)}\label{fig:1BasicCoop}
\end{minipage}
\begin{minipage}[t]{0.325 \linewidth}
\centering
\includegraphics[width=1\textwidth]{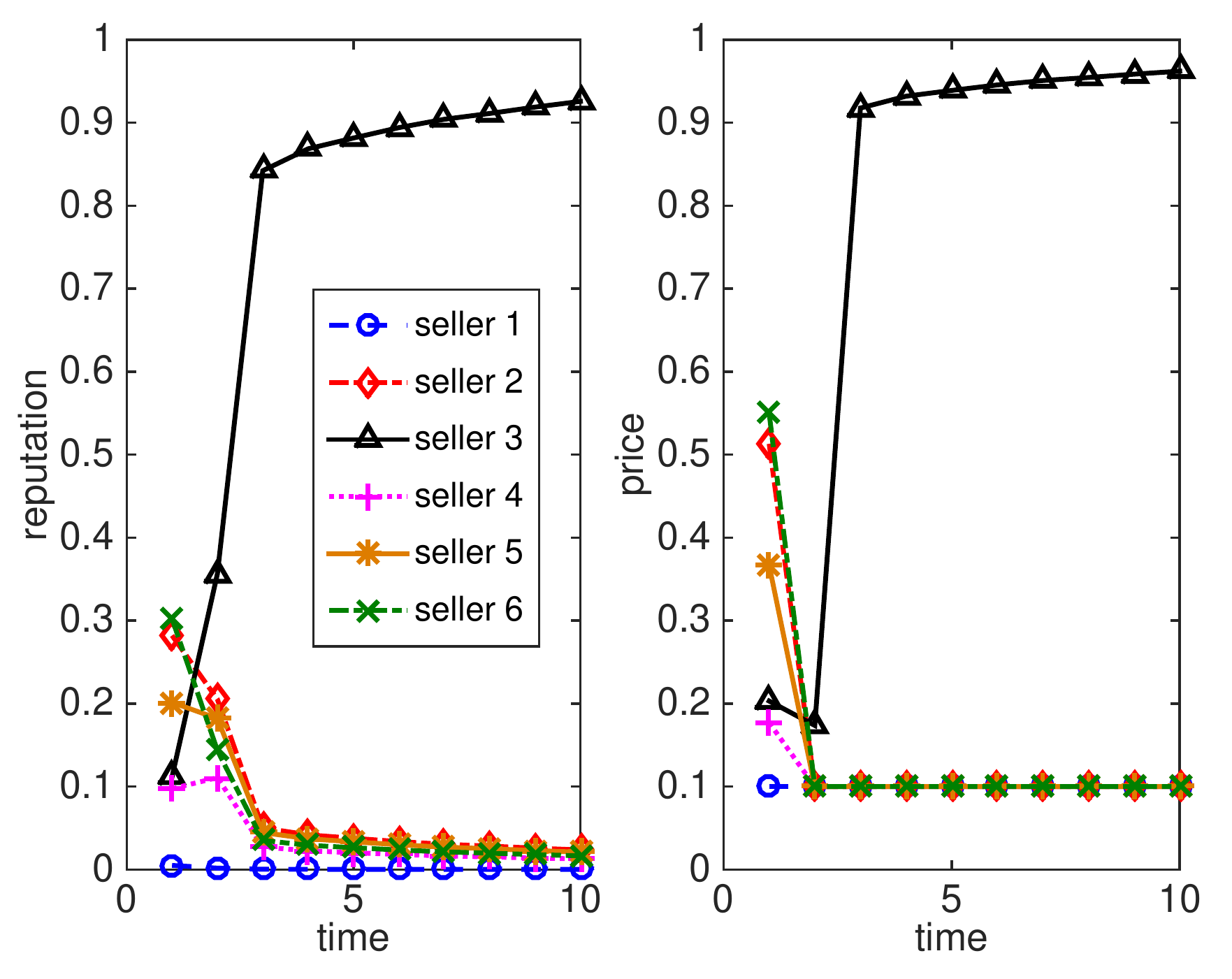}
\vspace{-3mm}
\caption{Reputation and Price Dynamics (Unlimited Capacity, $\boldsymbol{X}'$ and $\boldsymbol{\omega}^{I}$)}\label{fig:1BasicX2W1}
\end{minipage}
%\vspace{-3mm}
\end{figure*}

\subsection{Dynamic Game Analysis}\label{subsec:DGLM}

Next we will discuss the dynamic game in the infinite time horizon. 
We first derive the profit maximum strategy that maximizes the sellers' total profit in a time slot $t$.

\begin{lemma}\label{lemma:coopLM}
There exists a unique multi-seller price profile $\boldsymbol{p}^{\rm{MS-LM}}[t]$ that maximizes the sellers' total profits in each time slot $t$: 
\begin{equation}\label{eq:pcoopLM}
p_s^{\rm{MS-LM}}[t]=p_s^{\max}(r_s[t])=\frac{\rho}{x_0+1/r_s[t]}, \forall s \in \S.
\end{equation}
\end{lemma}

\begin{proof}
See Appendix \ref{app:lemma_coopLM}. % in the supplementary material.
\end{proof}

Similar to Section \ref{subsec:DGLC}, the multi-seller strategy can be enforced as part of the SPNE of the dynamic game by a punishment strategy. 
Due to limitation of space, we show details of the SPNE analysis in the supplementary material.

We summarize a comparison of the results of Sections \ref{sec:unlimited}, \ref{sec:limitedone}, and \ref{sec:limitedmulti} in Table \ref{table:comparison}.

\section{Simulation Results}\label{sec:simu}

We have performed simulations to illustrate the dynamics of sellers' reputations and prices, as well as the impact of system parameters, such as the initial transaction histories and ratings, on the reputation evolution process.

We have collected some data from the Airbnb website and set our simulation parameters accordingly. 
We focus here on the houses that are available in the neighbourhood Chuo-ku of Osaka, Japan, which is the most popular location to visit in Japan in 2016 \cite{Airbnb10Location}. 
Totally there are only $S=6$ houses in Chuo-ku provided by Airbnb hosts since most Airbnb hosts provide apartments. 
We have carried out simulations corresponding to the different modeling choices in Sections \ref{sec:unlimited}, \ref{sec:limitedone}, and \ref{sec:limitedmulti}, with different assumptions of the practices.\footnote{For example, if a single house has many beds, each bed can serve one buyer, and the market demand is small (as most travelers will rent apartments instead of houses), we can approximately view each house as having unlimited capacity. 
On the other hand, we can model a house in the one buyer per seller scenario by limiting the service to one buyer (or one family) a day, or model the house in the multiple buyers per seller scenario by allowing the house to serve multiple buyers (or families) a day.} 

We recorded the number of completed transactions of these 6 houses on Airbnb from the time they joined Airbnb up to June 2016, by exploring the completed transactions between hosts and guests, i.e., $\boldsymbol{X}=\{60,20,29,7,21,3\}$, and treated these as the initial transaction history in our simulation. 
To understand the impact of transaction history, we have also performed simulations by considering a different possibility of the initial numbers of completed transactions of these 6 houses to be $\boldsymbol{X}'=\{3,21,7,29,20,60\}$. 

To examine the impact of the fixed initial buyer ratings, we considered two different initial buyer ratings of these 6 houses. 
\begin{itemize}
\item $\boldsymbol{\omega}^{I}=\{0.1,0.8,0.95,0.2,0.6,0.3\}$,
\item $\boldsymbol{\omega}^{II}=\{0.45,0.8,0.95,0.2,0.6,0.3\}$. 
\end{itemize}

We assume that the marginal cost is $c=0.1$ \$/day, buyers' evaluation parameter is $\rho=10$, and the small positive number $\varepsilon=10^{-6}$. 
We assume that the buyers' arrival process follows a Poisson distribution with an arrival rate of $\lambda=10/$day.

\begin{figure*}[t]
%\vspace{-3mm}
 \centering
\begin{minipage}[t]{0.325 \linewidth}
\centering
\includegraphics[width=1\textwidth]{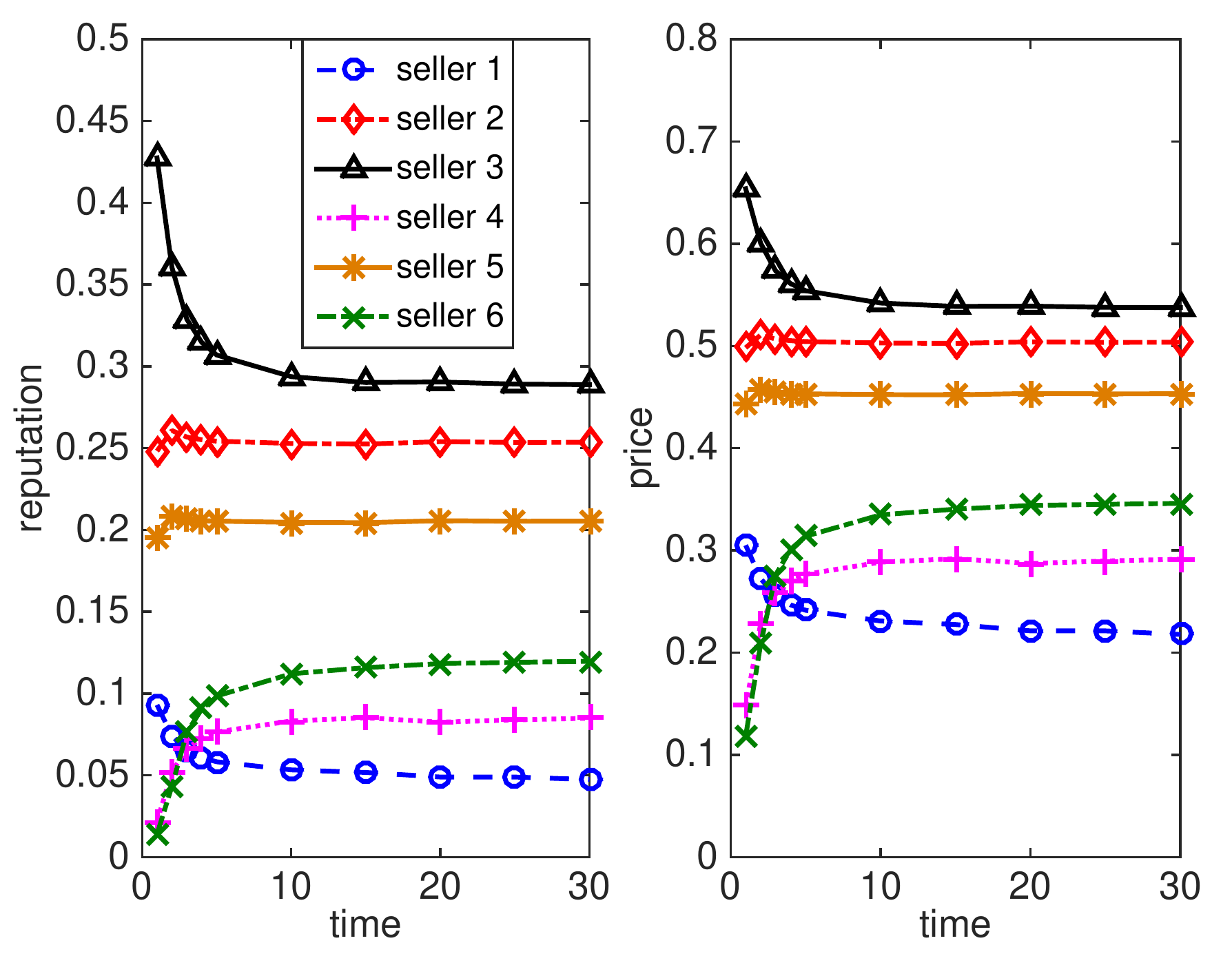}
\vspace{-3mm}
\caption{Reputation and Price Dynamics (Limited Capacity, One Buyer Per Seller, $\boldsymbol{X}$ and $\boldsymbol{\omega}^{I}$)}\label{fig:2UniCapX1W1}
\end{minipage}
\begin{minipage}[t]{0.325 \linewidth}
\centering
\includegraphics[width=1\textwidth]{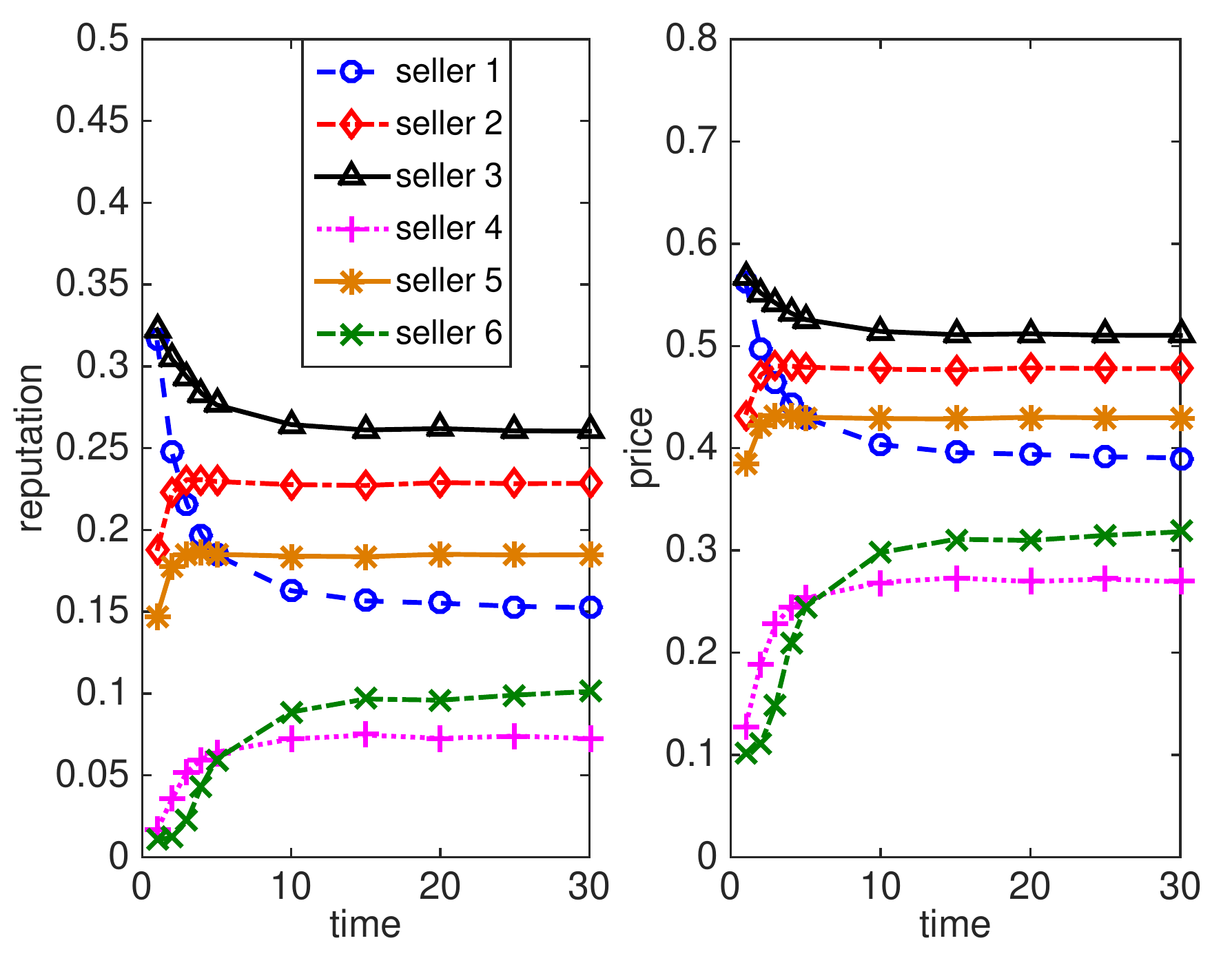}
\vspace{-3mm}
\caption{Reputation and Price Dynamics (Limited Capacity, One Buyer Per Seller, $\boldsymbol{X}$ and $\boldsymbol{\omega}^{II}$)}\label{fig:2UniCapX1W2}
\end{minipage}
\begin{minipage}[t]{0.325 \linewidth}
\centering
\includegraphics[width=1\textwidth]{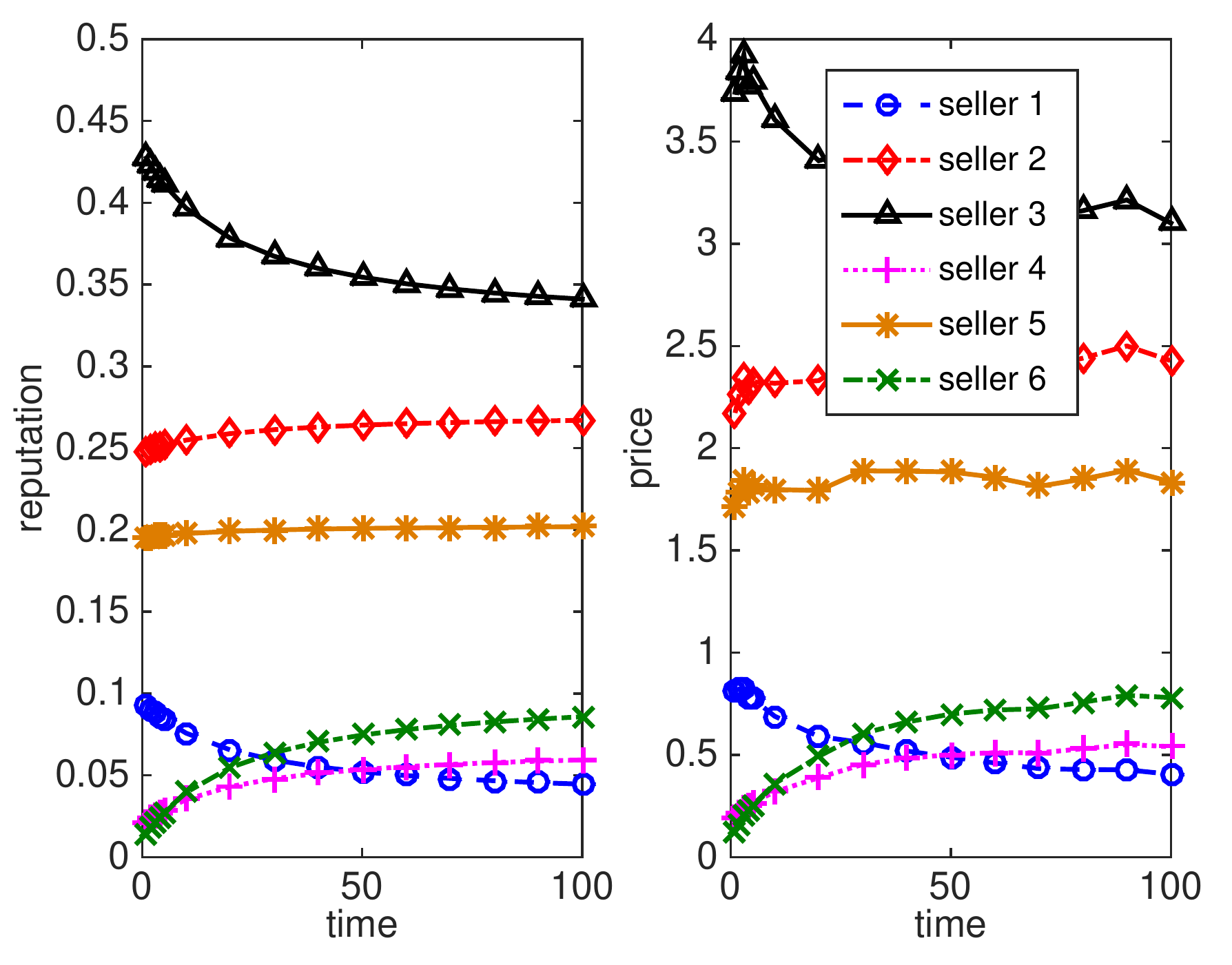}
\vspace{-3mm}
\caption{Reputation and Price Dynamics (Limited Capacity, Multiple Buyers Per Seller, $\boldsymbol{X}$ and  $\boldsymbol{\omega}^{I}$)}\label{fig:3RanCap}
\end{minipage}
%\vspace{-3mm}
\end{figure*}

\subsection{Unlimited Capacity Scenario}\label{simu:unlimited}

We first illustrate the dynamics of reputations and prices under the SPNE over time for sellers with unlimited capacity.

Figure \ref{fig:1Basic} and Figure \ref{fig:1BasicCoop} show the dynamics of sellers' reputations and prices under the SPNE over time under the same transaction history $X$ but different buyer ratings $\boldsymbol{\omega}^{I}$ and $\boldsymbol{\omega}^{II}$.

In Figure \ref{fig:1Basic} under $\boldsymbol{\omega}^{I}$, seller $3$ has the highest initial reputation that is much larger than that of anyone else (as shown in the left subfigure). 
As time goes by, a monopoly market emerges as indicated by Theorem \ref{theo:NErNE}: Seller $3$ charges a high price and dominates the market, while all other sellers charge according to their cost $c$ and make zero profit.

In Figure \ref{fig:1BasicCoop} under $\boldsymbol{\omega}^{II}$, the initial reputations of the 6 sellers are $\{0.316,0.187,0.322,0.016,0.148,0.011\}$, i.e., sellers $1,2,3$ and $5$ have relatively high initial reputations, hence they choose the multi-seller strategy in the first $3$ time slots as indicated in Theorem \ref{theo:NEcoop}. 
As time goes by, seller $3$ accumulates a large enough reputation, and hence the process leads to monopoly as indicated in Theorem \ref{theo:NErNE}. 
The change from a multi-seller market to a monopoly market is due to the fact that sellers are bounded rational, i.e., they have difficulty in predicting the reputation dynamics in the long run, and hence they compute their pricing decisions (to maximize their long-term payoff) assuming that the reputations do not change.

To better understand the impact of $\boldsymbol{X}$ on sellers' reputations and prices under the stable state, we have performed simulations for sellers under $\boldsymbol{X}'$ and $\boldsymbol{\omega}^{I}$, which are shown in Figure \ref{fig:1BasicX2W1}. 
Compared to the result under $\boldsymbol{X}$ in Figure \ref{fig:1Basic}, we can see that although seller $3$ has a low initial reputation under $\boldsymbol{X}'$ in Figure \ref{fig:1BasicX2W1} due to a small initial transaction number, seller $3$ can still dominate the market finally due to his large rating. 
Hence, sellers' ratings $\boldsymbol{\omega}$ play a more important role in determining sellers' reputations and prices under the SPNE at the stable state.

\subsection{Limited Capacity and One Buyer Per Seller Scenario}\label{simu:limitedone}

Next we illustrate the dynamics of sellers' reputations and prices under the SPNE in the limited capacity and one buyer per seller scenario.

Figure \ref{fig:2UniCapX1W1} and Figure \ref{fig:2UniCapX1W2} show the dynamics of sellers' equilibrium reputations and prices in each time slot with the same transaction history and under different buyers' ratings $\boldsymbol{\omega}^{I}$ and $\boldsymbol{\omega}^{II}$. 
In both cases, the market converges to a state where every seller has a positive market share as in Theorem \ref{theo:RandCLOSPNE}. 
Furthermore, the order of sellers' reputations at the stable state is consistent with the order of their ratings in both cases.

Comparing Figure \ref{fig:1Basic} and Figure \ref{fig:2UniCapX1W1}, we can see that in the unlimited capacity scenario (Figure \ref{fig:1Basic}), the market will evolve to be a monopoly market, while in the limited capacity and one buyer per seller scenario (Figure \ref{fig:2UniCapX1W1}), the market will evolve to be a multi-seller market.

\subsection{Limited Capacity and Multiple Buyers Per Seller Scenario}\label{simu:limitedmultiple}

Finally we illustrate the dynamics of sellers' reputations and prices under SPNE in the limited capacity and multiple buyers per seller scenario. 
Due to limitation of space, we show only the dynamics under $\boldsymbol{X}$ and $\boldsymbol{\omega}^{I}$ in Figure \ref{fig:3RanCap}. 
As we can see, the market converges to a stable state where the order of sellers' reputations is consistent with the order of their ratings.

Comparing Figure \ref{fig:2UniCapX1W1} and Figure \ref{fig:3RanCap}, we can see that when sellers can serve multiple buyers (Figure \ref{fig:3RanCap}), the prices under the SPNE at the stable state are higher than those when sellers can serve only one buyer due to the flexibility of service.

\section{Conclusion}\label{sec:conc}
In this paper, we analyzed the buyer purchase and seller pricing decisions in online markets, taking into account the impact of seller reputation and market competition. 
In the unlimited capacity scenario, if the gap between sellers' highest and second highest reputations is large enough, then the highest reputation seller dominates the market. 
If sellers' reputation levels are relatively close, then sellers with relatively high reputations will survive at the equilibrium. 
Furthermore, the market can evolve over time and change from a multi-seller market to a monopoly market. 
In the limited capacity scenario, the monopoly market will not exist due to the limited capacity constraint, and hence only multi-seller markets with different multi-seller strategies exist.
Simulation results show that the dynamics of sellers' reputations and prices at the SPNE will converge to stable states. 
There are several interesting future research directions to explore. 
For example, one can consider more general reputation formulations and heterogeneous buyers with different product evaluations. 
It is also possible to analyze the competition caused by new entrants and study how new sellers accumulate their early reputations.

\newpage

\appendices

\section{Proof of Lemma \ref{lemma:Optx1}}\label{app:lemma_Optx1}

The proof is based on the concavity of the payoff function. 
We have:
\begin{align}
& v_s'(x_s)=\frac{\rho r_s}{1+r_sx_s}-p_s,\notag\\
& v_s''(x_s)=-\frac{\rho r_s^2}{(1+r_sx_s)^2}<0.\notag
\end{align}
Hence, the buyer's payoff maximization problem is strictly concave, thus admitting a unique maximum. 
If the maximum sulution $x_s^{\ast}$ is positive, then it can be calculated by setting the first order derivative to zero:
\begin{equation*}%\label{eq:BuyerFOC}
v_s'(x_s^\ast)=\frac{\rho r_s}{1+r_sx_s^\ast}-p_s=0,
\end{equation*}
which leads to
\begin{equation*}%\label{eq:BuyerOptiSolu}
x_s^\ast=\frac{\rho}{p_s}-\frac{1}{r_s},
\end{equation*}
assuming that this quantity is nonnegative. 
Otherwise, the solution will have to be on the boundary, $x_s^{\ast}=0$, which leads to \eqref{eq:Optx1}.

\section{Proof of Lemma \ref{lemma:NES1}}\label{app:lemma_NES1}

Lemma \ref{lemma:NES1} states that the unique Nash equilibrium of Game \ref{SCgame} is $(p_s^{\ast}=c: \forall s \in \S)$ if more than one seller has the highest level of reputation. 
We first prove that the strategy profile $(p_s^{\ast}=c, \forall s \in \S)$ is a Nash equilibrium, and then prove that it is the unique equilibrium. 
Recall that buyers will choose the seller with the highest reputation-price ratio, i.e., $\max_{s\in\S}\frac{r_s}{p_s}$. 

Step 1: we prove that $(p_s^{\ast}=c, \forall s \in \S)$ is a Nash equilibrium. 
If $p_s^{\ast}=c, \forall s \in \S$, no seller will earn a positive profit. 
However, if $p_{-s}=c$, then if seller $s$ deviates from the price profile by setting $p_s > c$, then it will not improve his profit, since all buyers will buy from the seller who has the highest reputation and still sets the competitive price, i.e., the seller with the highest reputation-price ratio. 
Hence no seller $s, ~\forall s \in \S,$ has an incentive to unilaterally increase his price to be larger than $c$ if other sellers set $p_{-s}=c$, and the strategy profile $(p_s^{\ast}=c, \forall s \in \S)$ is a Nash equilibrium.

Step 2: we prove that no other price vector is an equilibrium. 
If all sellers set the same reputation-price ratio, i.e., $\frac{r_s}{p_s}$, such that the price is above the marginal cost, and share the market, then each seller has an incentive to undercut the others by an arbitrarily small amount of decrease on price and capture the whole market and increase its profits. 
So there can be no equilibrium with all sellers setting the same reputation-price ratio such that the prices are above the marginal cost. 
Also, there can be no equilibrium with sellers setting different reputation-price ratios and the prices are above the marginal cost. 
The sellers setting the smaller ratio will earn nothing (the seller with the largest ratio serves all the buyers). 
Hence the seller with a smaller ratio will want to decrease its price to undercut the seller with a larger ratio until the price is decreased to equal the marginal cost. 

Hence the only equilibrium occurs when all sellers set their prices equal to the marginal cost.

\section{Proof of Lemma \ref{lemma:NES2}}\label{app:lemma_NES2}

Lemma \ref{lemma:NES2} states that when only one seller has the highest reputation, the unique Nash equilibrium of Game \ref{SCgame} is
\begin{equation}\label{eq:NEproof}
p_{s}^\ast=
\left\{
\begin{aligned}
\min \left \{ \frac{r_{\rm{max}}}{r_{\rm{sec}}}c-\varepsilon, \sqrt{c \rho r_s} \right \}, &\quad \mbox{ if } r_s = r_{\rm{max}};
\\
c, & \quad \mbox{ otherwise}.
\end{aligned}
\right.
\end{equation}
First we prove that the price profile in \eqref{eq:NEproof} is a Nash equilibrium. 
We denote the seller who has the highest reputation $r_{\rm{max}}$ as seller $s^{\rm max}$. 
If $p_{-s^{\rm max}}=c$ and seller $s^{\rm max}$ sets $p_{s^{\rm max}} < \min \left \{ \frac{r_{\rm{max}}}{r_{\rm{sec}}}c-\varepsilon, \sqrt{c \rho r_s} \right \}$, then the profit of seller $s^{\rm max}$ will decrease due to the formulation of its profit; if $p_{s^{\rm max}} > \min \left \{ \frac{r_{\rm{max}}}{r_{\rm{sec}}}c-\varepsilon, \sqrt{c \rho r_s} \right \}$, then its reputation-price ratio will decrease and it will lose buyers, leading to decrease in its profit. 
If seller $s ~(s \neq s^{\rm max})$ sets $p_s > c$, then it will earn nothing, since all buyers will buy from seller $s^{\rm max}$ with the highest reputation, i.e., the seller with the highest reputation-price ratio. 
Then we can see that the equilibrium price profile in \eqref{eq:NEproof} is unique, and no other price constitutes an equilibrium, using a proof similar to that in Appendix \ref{app:lemma_NES1}.

\section{Proof of Theorem \ref{theo:NErNE}}\label{app:theo_NErNE}

If $r_{\rm{max}}[t] > \frac{\rho (r_{\rm{sec}}[t]) ^2}{c}$, we have $\sqrt{c \rho r_{\rm{max}}[t]}  < \frac{r_{\rm{max}}[t] }{r_{\rm{sec}}[t] }c$. 
Hence the highest reputation seller will set the monopoly price at the NE of the single-slot game, i.e., $p_{s^{\rm max}[t]}^{\ast}[t]=\sqrt{c \rho r_{\rm{max}}[t]}$, and can attract all buyers at Nash equilibrium of the single-slot game and its profit under the Nash equilibrium of the single-slot game is larger than the profit under any other strategy. 
This seller will always play the Nash equilibrium of the single-slot game and has no incentive to play any other strategy. 
Hence, the repetition of the Nash equilibrium of the single-slot game is the unique subgame perfect Nash equilibrium of the repeated competition game.

\section{Proof of Theorem \ref{theo:RepEvoMono}}\label{app:theo_RepEvoMono}

We denote the reputation difference between two time slots $t+1$ and $t$ as $\Delta r_s[t]$, i.e., 
\begin{equation*}
\Delta r_s[t] \eq r_s[t+1]-r_s[t]=\frac{\omega_s \left(X_s[t]+q_{s}^\ast[t]\right)}{ \sum_{i\in\S} \left(X_i[t]+q_{i}^\ast[t]\right)}-\frac{\omega_s X_s[t]}{ \sum_{i\in\S} \omega_i X_i[t]}.
\end{equation*}
To prove the convergence, we need to prove $\lim_{t\rightarrow \infty} \Delta r_s[t] = 0$, i.e., 
$$\frac{\omega_s \left(X_s[t]+q_{s}^\ast[t]\right)}{ \sum_{i\in\S} \left(X_i[t]+q_{i}^\ast[t]\right)}=\frac{\omega_s X_s[t]}{ \sum_{i\in\S} \omega_i X_i[t]}, \mbox { for }t \to \infty ,$$
which is equivalent to
\begin{equation*}
\frac{\omega_s q_{s}^\ast[t]}{ \sum_{i\in\S} q_{i}^\ast[t]}=\frac{\omega_s X_s[t]}{ \sum_{i\in\S} \omega_i X_i[t]}=r_s[t], \mbox { for }t \to \infty.
\end{equation*}

In a monopoly market, if seller $s$ is not the highest reputation seller, i.e., $r_s[t] \neq r_{\max}[t]$, then $q_s^\ast[t]=0$ according to \eqref{eq:qsNES} and \eqref{eq:pNES}. 
So the reputation of seller $s$ monotonically decreases, and $r_s[\infty] \to 0$. 
Hence, 
$$\frac{\omega_s q_{s}^\ast[\infty]}{ \sum_{i\in\S} q_{i}^\ast[\infty]}=r_s[\infty]=0.$$

If seller $s$ is the highest reputation seller, i.e., $r_s[t] = r_{\max}[t]$, then seller $s$ is the monopoly in the market and
$$\frac{\omega_s q_{s}^\ast[t]}{ \sum_{i\in\S} q_{i}^\ast[t]}=1.$$
The number of completed transactions and the reputation of seller $s$ monotonically increase, and 
$r_s[\infty] = \frac{\omega_s X_s[\infty]}{ \sum_{i\in\S} \omega_i X_i[\infty]} \to 1$. 
Hence, 
$$\frac{\omega_s q_{s}^\ast[\infty]}{ \sum_{i\in\S} q_{i}^\ast[\infty]}=r_s[\infty]=1.$$

In summary, the reputation evolution process of the monopoly market converges.

\section{Proof of Lemma \ref{lemma:coop}}\label{app:lemma_coop}

% \begin{proof}

Under the multi-seller strategy, i.e., sellers have the same ratio $\varphi=\frac{r_s}{p_s}$ and buyers choose each seller with the same probability, seller $s$'s profit from a buyer is:
\[ \pi_s(\varphi_s)= (\frac{1}{\varphi_s} - \frac{c}{r_s})(\rho \varphi_s-1).\]
For seller $s$, the ratio $\varphi_s$ that maximizes his profit $\pi_s$ is
\[\varphi_s^{\ast}=\sqrt{\frac{r_s}{c \rho }}.\]
If $\varphi_s < \varphi_s^{\ast} $, seller $s$'s profit increases with $\varphi_s$, and if $\varphi_s > \varphi_s^{\ast} $, seller $s$'s profit decreases with $\varphi_s$. 

Define the largest reputation by $r_{\rm{max}}=\max_{s\in\S}r_s$.

We claim that the ratio
\[\varphi^{\rm{MS-U}}=\sqrt{\frac{r_{\rm{max}}}{c \rho}}\]
is the unique one under which sellers are willing to play the multi-seller strategy and sellers' joint profits are maximized. 

\textbf{Case I:} If $\varphi > \varphi^{\rm{MS-U}}$, sellers' profits will be smaller under $\varphi$ than the profit under $\varphi^{\rm{MS-U}}$ due to the structure of the profit function.

\textbf{Case II:} If $\varphi < \varphi^{\rm{MS-U}}$, the seller $s^{\rm max}$ with the highest reputation $r_{\rm{max}}$ has incentive to increase his ratio to be $\varphi_{s^{\rm max}} = \varphi^{\rm{MS-U}}$ such that his profit obtained from a single buyer is increased and he would be attractive to all buyers since he has the highest ratio $\varphi_{s^{\rm max}}$ among all sellers. 

In summary, $\varphi^{\rm{MS-U}}$ is the unique ratio under which sellers are willing to play the multi-seller strategy and sellers' joint profits are maximized. 

Note that if for seller $s$, $\frac{r_s}{\varphi^{\rm{MS-U}}}<c$, then seller $s$ is not able to earn positive profit under $\varphi^{\rm{MS-U}}$. 
Hence, he will simply set his price to equal the marginal cost $c$.

Hence, the price strategy $p_s^{\rm{MS-U}}=\max \left \{ \frac{r_s}{\varphi^{\rm{MS-U}}}, c \right \}, \forall s \in \S $, is the unique price strategy under which sellers are willing to play the multi-seller strategy and sellers' joint profits are maximized. 

% (Q.E.D.)
% \end{proof}

\section{Proof of Corollary \ref{coro:rth}}\label{app:coro_rth}

% \begin{proof}
Seller $s$ is able to achieve strictly positive profit under the multi-seller strategy if and only if the reputation-price ratio is smaller than the largest ratio that seller $s$ can set, i.e.,
\[ \varphi^{\rm{MS-U}}=\sqrt{\frac{r_{\rm{max}}}{c \rho}} < \frac{r_s}{c} ,\]
which leads to
\[r_s > r_{\rm{th}} \eq \sqrt{\frac{cr_{\rm{max}}}{\rho}}.\]
% \end{proof}

\section{Proof of Theorem \ref{theo:NEcoop}}\label{app:theo_NEcoop}

If $\frac{\pi_{s^{\max}[t]}^{\rm{MS-U}}[t]}{S_C[t]} \geq \frac{\pi_{s^{\max}[t]}^{\rm NE-U}[t]}{S_L[t]}$, sellers have incentive to play the multi-seller strategy to achieve a higher profit than the profit under the Nash equilibrium of the single-slot game. 
Since sellers are myopic and play the dynamic game by assuming that all sellers' future reputations remain the same as the current reputations, for any $t' \geq t$, we have $\pi_s^{\rm{MS-U}}[t']=\pi_s^{\rm{MS-U}}[t]=\pi_s^{\rm{MS-U}}, S_C[t']=S_C[t]=S_C, \pi_s^{\rm{NE-U}}[t']=\pi_s^{\rm{NE-U}}[t]=\pi_s^{\rm{NE-U}}, S_L[t']=S_L[t]=S_L$.
By the one-step deviation principle, the multi-seller strategy is a subgame-perfect Nash equilibrium if: 
\[ \sum_{t=0}^{\infty} \delta^t \frac{N}{S_C} \pi_s^{\rm{MS-U}} \geq N \pi_s^{\rm{MS-U}} + \sum_{t=1}^{\infty} \delta^t \frac{N}{S_L} \pi_s^{\rm{NE-U}} .\]
The left-hand side of the above equation is the profit of seller $s$ if he plays according to the multi-seller strategy. 
The right-hand side of the above equation is the profit of seller $s$ if he deviates from the multi-seller strategy. 
By solving the above equation, we can arrive at the following condition on the discount factor:
\[\delta_s > \frac{S_C-1}{S_C} \cdot \frac{\pi_s^{\rm{MS-U}}}{\pi_s^{\rm{MS-U}} - {\pi_s^{\rm{NE-U}}}/{S_L}}, \forall s \in \S. \]
Hence, for each time slot $t$, we have the corresponding condition:
$$\delta_s > \frac{{S}_C[t]-1}{{S}_C[t]} \cdot \frac{{\pi}_s^{\rm{MS-U}}[t]}{{\pi}_s^{\rm{MS-U}}[t] - {\pi}_s^{\rm NE-U}[t]/{S}_L[t]}, \forall s \in S.$$

\section{Proof of Lemma \ref{lemma:3seller}}\label{app:lemma_3seller}

We first prove Case I where $$\pi_2^{\rm mon} \geq \pi_2^{\rm dev}, \pi_3^{\rm mon} \geq \pi_3^{\rm dev}.$$
Toward that end, we first prove that the price profile $\boldsymbol{p}^{\ast \rm NE-LO}$ which is 
$$p_s^{\ast \rm NE-LO}=\sqrt{c \rho r_s}, \forall s=1,2,3,$$ 
is a Nash equilibrium. 
Under the price profile $\boldsymbol{p}^{\ast \rm NE-LO}$, the reputation-price ratios satisfy $\frac{r_1}{p_1^{\ast \rm NE-LO}} > \frac{r_2}{p_2^{\ast \rm NE-LO}}> \frac{r_3}{p_3^{\ast \rm NE-LO}}$ and sellers' profits are
\begin{equation}\label{eq:piNEapp}
\begin{aligned}
& \pi_1(p_1^{\ast \rm NE-LO},p_{-1}^{\ast \rm NE-LO})= (p_1^{\ast \rm NE-LO}-c)\left(\frac{\rho}{p_1^{\ast \rm NE-LO}}-\frac{1}{r_1}\right) \cdot  \\
&~~~~~~~~~~~~~~~~~~~~~~~~~~~~~~ \left(1-P(k=0)\right) ,   \\
& \pi_2(p_2^{\ast \rm NE-LO},p_{-2}^{\ast \rm NE-LO})= (p_2^{\ast \rm NE-LO}-c)\left(\frac{\rho}{p_2^{\ast \rm NE-LO}}-\frac{1}{r_2}\right)  \cdot  \\
&~~~~~~~~~~~~~~~~~~~~~~~~~~~~~~  \left(1-P(k=0)-P(k=1)\right) ,  \\
& \pi_3(p_3^{\ast \rm NE-LO},p_{-3}^{\ast \rm NE-LO})= (p_3^{\ast \rm NE-LO}-c)\left(\frac{\rho}{p_3^{\ast \rm NE-LO}}-\frac{1}{r_3}\right) \cdot   \\
&~~~~~~~~~~~~~~~~~~~~~~~~~~~~~~  (1-P(k=0)-P(k=1)-P(k=2) ) . 
\end{aligned}
\end{equation}
Since $1-\sum_{k=0}^i P(X=k)$ is independent of prices, the monopoly prices $$p_s^{\ast \rm NE-LO}=\sqrt{c \rho r_s}, \forall s=1,2,3,$$  can maximize sellers' profits. 
Since none of the three sellers can achieve a higher profit by deviating from their monopoly prices, the price profile $\boldsymbol{p}^{\ast \rm NE-LO}$ is a Nash equilibrium. 
Then we prove that no other price profile is an equilibrium. 
If sellers set their prices such that their reputation-price ratios satisfy $\frac{r_1}{p_1} > \frac{r_2}{p_2} > \frac{r_3}{p_3}$ but at least one seller does not set his monopoly price, then this seller will change his price to his monopoly price so that he can achieve a higher profit. 
If sellers set their prices such that their reputation-price ratios do not satisfy $\frac{r_1}{p_1} > \frac{r_2}{p_2} > \frac{r_3}{p_3}$, then seller $1$ will decrease his price such that his reputation-price ratio is the highest to achieve a higher profit, and seller $2$ will decrease his price such that his reputation-price ratio is higher than the reputation-price ratio of seller $3$ to achieve a higher profit. 
In summary, the price profile $\boldsymbol{p}^{\ast \rm NE-LO}$ is the unique Nash equilibrium.

Now we prove Case II where $$\pi_2^{\rm mon} < \pi_2^{\rm dev} \pi_3^{\rm mon} \geq \pi_3^{\rm dev}.$$ 
Toward that end, we first prove that the price profile $\boldsymbol{p}^{\ast \rm NE-LO}$ which is 
\begin{equation}\label{eq:3sellerNE2app}
\begin{aligned}
&p_1^{\ast \rm NE-LO} = \max \left \{  \min \left\{ \sqrt{c \rho r_1}, \frac{r_1}{r_2}p_2^{\ast \rm NE-LO}-\varepsilon \right\}, p_1^{\min}  \right \},\\
&  p_2^{\ast \rm NE-LO} = \max \left\{ \min \left\{ \sqrt{c \rho r_2}, \frac{r_2}{r_3}p_3^{\ast \rm NE-LO}-\varepsilon \right\}, p_2^{\min} \right\}, \\
&  p_3^{\ast \rm NE-LO} = \sqrt{c \rho r_3},
\end{aligned}
\end{equation} 
is a Nash equilibrium. 
Under the price profile $\boldsymbol{p}^{\ast \rm NE-LO}$, the reputation-price ratios satisfy $\frac{r_1}{p_1^{\ast \rm NE-LO}} > \frac{r_2}{p_2^{\ast \rm NE-LO}}> \frac{r_3}{p_3^{\ast \rm NE-LO}}$ and sellers' profits can be calculated as in \eqref{eq:piNEapp}. 
Since none of the three sellers can achieve a higher profit by deviating from his current price, the price profile $\boldsymbol{p}^{\ast \rm NE-LO}$ is a Nash equilibrium. 
Then we prove that no other price profile is an equilibrium. 
If sellers set their prices such that their reputation-price ratios satisfy $\frac{r_1}{p_1} > \frac{r_2}{p_2} > \frac{r_3}{p_3}$ but at least one seller $s$ sets $p_s \neq p_s^{\ast \rm NE-LO}$, then seller $3$ will change his price to his monopoly price to achieve the highest profit under $\frac{r_1}{p_1} > \frac{r_2}{p_2} > \frac{r_3}{p_3}$, while seller $1$ and seller $2$ will compete and change their prices to the competitive price as in \eqref{eq:3sellerNE2app}. 
If sellers set their prices such that their reputation-price ratios do not satisfy $\frac{r_1}{p_1} > \frac{r_2}{p_2} > \frac{r_3}{p_3}$, then seller $1$ will decrease his price such that his reputation-price ratio is the highest to achieve a higher profit, and seller $2$ will decrease his price such that his reputation-price ratio is the higher than the reputation-price ratio of seller $3$ to achieve a higher profit. 
In summary, the price profile $\boldsymbol{p}^{\ast \rm NE-LO}$ is the unique Nash equilibrium.

Finally we prove Case III where $\pi_3^{\rm mon} < \pi_3^{\rm dev}.$ 
Toward that end, we first prove that the price profile $\boldsymbol{p}^{\ast \rm NE-LO}$ which is 
\begin{equation}\label{eq:3sellerNE3app}
\begin{aligned}
& p_1^{\ast \rm NE-LO} = \max \left\{  \min \left\{ \sqrt{c \rho r_1}, \frac{r_1}{r_2}p_2^{\ast \rm NE-LO}-\varepsilon \right\}, p_1^{\min}  \right\},  \\
&p_2^{\ast \rm NE-LO} = \max \left \{  \min \left\{ \sqrt{c \rho r_2}, \frac{r_2}{r_3}p_3^{\ast \rm NE-LO}-\varepsilon \right\},  p_2^{\min}  \right\}, \\
& p_3^{\ast \rm NE-LO} = p_3^{\min}.
\end{aligned}
\end{equation}
is a Nash equilibrium. 
Under the price profile $\boldsymbol{p}^{\ast \rm NE-LO}$, the reputation-price ratios satisfy $\frac{r_1}{p_1^{\ast \rm NE-LO}} > \frac{r_2}{p_2^{\ast \rm NE-LO}}> \frac{r_3}{p_3^{\ast \rm NE-LO}}$ and sellers' profits can be calculated as in \eqref{eq:piNEapp}. 
Since none of the three sellers can achieve a higher profit by deviating from his current price, the price profile $\boldsymbol{p}^{\ast \rm NE-LO}$ is a Nash equilibrium. 
Then we prove that no other price profile is an equilibrium. 
If sellers set their prices such that their reputation-price ratios satisfy $\frac{r_1}{p_1} > \frac{r_2}{p_2} > \frac{r_3}{p_3}$ but at least one seller $s$ sets $p_s \neq p_s^{\ast \rm NE-LO}$, then seller $1$, seller $2$, and seller $3$ will compete and change their prices to the competitive price as in \eqref{eq:3sellerNE3app}. 
If sellers set their prices such that their reputation-price ratios do not satisfy $\frac{r_1}{p_1} > \frac{r_2}{p_2} > \frac{r_3}{p_3}$, then seller $1$ will decrease his price such that his reputation-price ratio is the highest to achieve a higher profit, and seller $2$ will decrease his price such that his reputation-price ratio is the higher than the reputation-price ratio of seller $3$ to achieve a higher profit. 
In summary, the price profile $\boldsymbol{p}^{\ast \rm NE-LO}$ is the unique Nash equilibrium.

\section{Proof of Theorem \ref{theo:RandCLONE}}\label{app:theo_RandCLONE}

If the price profile derived by Algorithm \ref{algo:SAA} is $p_s^\ast=\sqrt{c \rho r_s}, ~\forall s \in \S$, we prove that it is the unique Nash equilibrium. 
According to Algorithm \ref{algo:SAA}, no seller has any incentive to lower his price. 
And monopoly price is the price that can maximize each seller's profit. 
Hence, no seller has any incentive to change his price, and the price profile is a Nash equilibrium. 
Next we prove that it is the unique Nash equilibrium. 
If a seller's price is larger than his monopoly price, the seller has an incentive to decrease his price to his monopoly price so that he can be more attractive to buyers due to a larger reputation-price ratio and he can obtain more profit from a single buyer by setting the monopoly price. 
If a seller's price is smaller than his monopoly price, the seller has an incentive to increase his price to his monopoly price due to the fact that the seller has no incentive to lower his price from the monopoly price.

For the case when the price profile derived by Algorithm \ref{algo:SAA} is not the monopoly price profile, the proof follows the lines of the proof of Theorem \ref{theo:NES}.

\section{Proof of Lemma \ref{lemma:cooplimitedone}}\label{app:lemma_cooplimitedone}

Under the price strategy $p_s^{\rm MS-LO}=\sqrt{c \rho r_s}, \forall s \in \S$, sellers' reputation-price ratio satisfies
\[ \varphi_i^{\rm MS-LO} > \varphi_j^{\rm MS-LO}, \mbox{ if } r_i>r_j ,\]
which is consistent with the result at Nash equilibrium of the single-slot game:
\[ \varphi_i^{\ast \rm{NE-LO}} > \varphi_j^{\ast \rm{NE-LO}}, \mbox{ if } r_i>r_j .\]
Also, under the price strategy $p_s^{\rm MS-LO}=\sqrt{c \rho r_s}, \forall s \in \S$, each seller $s$ sets his monopoly price and can achieve the maximum profit from each buyer as the profit in a monopoly market.

\section{Proof of Theorem \ref{theo:RandCLOSPNE}}\label{app:theo_RandCLOSPNE}

Since sellers are myopic and play the dynamic game by assuming that all sellers' future reputations remain the same as the current reputations, for any $t' \geq t$, we have $\pi_s^{\rm{MS-LO}}[t']=\pi_s^{\rm{MS-LO}}[t]=\pi_s^{\rm{MS-LO}}, \pi_s^{\rm{NE-LO}}[t']=\pi_s^{\rm{NE-LO}}[t]=\pi_s^{\rm{NE-LO}}$.
By the one-step deviation principle, the multi-seller strategy is a subgame-perfect Nash equilibrium if: 
\[ \sum_{t=0}^{\infty} \delta^t \pi_s^{\rm{MS-LO}} \geq \pi_s^{\rm{dev}} + \sum_{t=1}^{\infty} \delta^t  \pi_s^{\rm{NE-LO}} .\]
The left-hand side of the above equation is the profit of seller $s$ if he plays according to the multi-seller strategy. 
The right-hand side of the above equation is the profit of seller $s$ if he deviates from the multi-seller strategy. 
By solving the above inequality, we arrive at the following condition on the discount factor:
\[\delta_s > \frac{ \pi_s^{\rm{dev}}-\pi_s^{\rm{MS-LO}}}{ \pi_s^{\rm{dev}}-\pi_s^{\rm{NE-LO}}} , \forall s \in \S.
 \]
Hence, for each time slot $t$, we have the corresponding condition:
\[\delta_s > \frac{ \pi_s^{\rm{dev}}[t]-\pi_s^{\rm{MS-LO}}[t]}{ \pi_s^{\rm{dev}}[t]-\pi_s^{\rm{NE-LO}}[t]} , \forall s \in \S.
 \]

\section{Proof of Theorem \ref{theo:RandCLM}}\label{app:theo_RandCLM}

The proof follows the lines of the proof of Theorem \ref{theo:RandCLONE}.

\section{Proof of Lemma \ref{lemma:coopLM}}\label{app:lemma_coopLM}

The proof follows the lines of the proof of Lemma \ref{lemma:cooplimitedone}.

\section{Subgame Perfect Nash equilibrium of Dynamic Seller Competition Game with Limited Capacity and Multiple Buyers Per Seller}

\begin{theorem}\label{theo:RandCLMSPNE}
The following holds. 
\begin{itemize}
\item If at the Nash equilibrium of the single-slot game, all sellers set the same reputation-price ratio, then the unique SPNE involves choosing the price profile according to the Nash equilibrium of the single-slot game in each time slot.

\item Otherwise, there exists a unique multi-seller price vector $\boldsymbol{p}^{\rm{MS-LM}}$ that maximizes the sellers' joint profits: 
\begin{equation}
p_s^{\rm{MS-LM}}=\max \left\{ \frac{r_s}{\varphi^{\rm{MS-LM}}} ,p_s^{\min} \right \}, \forall s \in \S 
\end{equation}
where $\varphi^{\rm{MS-LM}}\eq \varphi_{th}^{\max}+ \varepsilon$ and $\varphi_{th}^{\max}$ is a chosen seller's largest reputation-price ratio that satisfies
\begin{align*}
&  \varphi_{th}^{\max}\geq \sqrt{\frac{r_{\max}}{c \rho}}, \\
& \frac{r_s}{p_s^{\rm{MS-LM}}}\leq \varphi_{th}^{\max}, \mbox{ if } \varphi_s^{\max} \leq \varphi_{th}^{\max},  \\
& \frac{r_s}{p_s^{\rm{MS-LM}}}=\varphi_{th}^{\max}+ \varepsilon, \mbox{ if } \varphi_s^{\max}>\varphi_{th}^{\max},  \\
& N^{\rm MS-LM} = \sum_{\{s:\frac{r_s}{p_s^{\rm{MS-LM}}}=\varphi_{th}^{\max}+ \varepsilon\}}n_s(p_s^{\rm MS-LM}), \\
& \frac{\sum_{k=0}^\infty kP(X=k)}{ N^{\rm MS-LM} }(p_s^{\rm{MS-LM}}-c) \left(  \frac{\rho}{p_s^{\rm{MS-LM}}}-\frac{1}{r_{s}} \right) \geq \\
&~~(p_s^\ast-c) \left(  \frac{\rho}{p_s^\ast}-\frac{1}{r_{s}} \right),\forall s \in \S.
\end{align*}
We let $\pi_s^{NE}[t]$ denote the profit of seller $s$, achieved at the NE of the single-slot game in time slot $t$ derived by Algorithm \ref{algo:RSAA}. 
We let $\pi_s^{\rm{MS-LM}}[t]$ denote owner $s$'s expected profit achieved under the multi-seller strategy in time slot $t$, and $\pi_s^{\rm{dev}}[t]$ denote the maximum expected profit that owner $s$ can achieve by unilaterally deviating from the multi-seller strategy. 
Consider the following strategy profile: all sellers set the price profile $\boldsymbol{p}[t]=\boldsymbol{p}^{\rm{MS-LM}}[t]$ in each time slot $t$ until a seller deviates, in which case all sellers choose the price profile according to the Nash equilibrium of the single-slot game for ever. 
Such a strategy profile is an SPNE if 
\begin{equation}\label{con:deltaCLMRan}
\delta_s > \frac{ \pi_s^{\rm{dev}}[t]-\pi_s^{\rm{MS-LM}}[t] }{ \pi_s^{\rm{dev}}[t]-\pi_s^{NE}[t]} , \forall s \in \S.
\end{equation}
\end{itemize}
\end{theorem}

The proof of Theorem \ref{theo:RandCLMSPNE} follows the lines of the proof of Theorem \ref{theo:RandCLOSPNE}.

\end{document}